\documentclass[a4paper,10pt, margin = 1in]{article}
\usepackage[utf8]{inputenc}

\usepackage[a4paper]{geometry}
\usepackage{soul}

\usepackage{amsthm}
\usepackage{graphicx}
\usepackage{subcaption}
\usepackage{color}
\usepackage[colorinlistoftodos]{todonotes}

\usepackage{natbib}

\newtheorem{theorem}{Theorem}
\newtheorem{lemma}[theorem]{Lemma}
\newtheorem{definition}[theorem]{Definition}
\newtheorem{conjecture}[theorem]{Conjecture}

\author{Leo van Iersel\footnote{Delft Institute of Applied Mathematics, Delft University of Technology, Van Mourik Broekmanweg 6,
2628 XE, Delft, The Netherlands, \{L.J.J.vanIersel, M.E.L.Jones\}@tudelft.nl. Research funded in part by the Netherlands Organization for Scientific Research (NWO), including Vidi grant 639.072.602, and partly by the 4TU Applied Mathematics Institute.} \and Mark Jones\footnotemark[1] \and Steven Kelk\footnote{Department of Data Science and Knowledge Engineering (DKE), Maastricht University, Bouillonstraat 8-10
6211 LH, Maastricht, The Netherlands, steven.kelk@maastrichtuniversity.nl}}

\title{A third strike against perfect phylogeny\thanks{This article has been accepted for publication in Systematic Biology Published by Oxford University Press.}}
\date{\today}

\begin{document}
\maketitle
\begin{abstract}
\noindent  
Perfect phylogenies are fundamental in the study of evolutionary trees because they capture the situation when each evolutionary trait emerges only once in history; if such events are believed to be rare, then by Occam's Razor such parsimonious trees are preferable as a hypothesis of evolution. A classical result states that 2-state characters permit a perfect phylogeny precisely if each subset of~2 characters permits one. More recently, it was shown that for 3-state characters the same property holds but for size-3 subsets. A long-standing open problem asked whether such a constant exists for each number of states. More precisely, it has been conjectured that for any fixed 
number of states $r$
there exists a constant $f(r)$ such that a set of $r$-state characters~$C$ has a perfect phylogeny if and only if every subset of at most~$f(r)$ characters has a perfect phylogeny. Informally, the conjecture states that checking fixed-size subsets of characters is enough to correctly determine whether input data permits a perfect phylogeny, irrespective of the number of characters in the input. In this paper, we show that this conjecture is false.
In particular, we show that for any constant~$t$, there exists a set~$C$ of $8$-state characters such that~$C$ has no perfect phylogeny, but there exists a perfect phylogeny for every subset of at most~$t$ characters. Moreover, there already exists a perfect phylogeny when ignoring just one of the characters, independent of which character you ignore.
This negative result complements the two negative results (``strikes'') of 
\citet{bodlaender1992two,bodlaender2000hardness}. We reflect on the consequences of this third strike, pointing out that while it does close off some routes for efficient algorithm development, many others remain open.

\end{abstract}


\setcounter{section}{1}
\setcounter{equation}{0}
The
traditional model for capturing the evolution of a set~$X$ of contemporary species or taxa is the \emph{phylogenetic tree.} In such trees  internal nodes represent hypothetical (common) ancestors. 
The central goal in phylogenetics is to infer
phylogenetic trees
given only data obtained from (or observed at)~$X$ e.g. DNA sequences, amino acid sequences or morphological features \citep{felsenstein2004inferring}. The data observed at a taxon~$x$ in $X$ is typically represented as an ordered length-$m$ vector of discrete \emph{states}, where the states are elements of some size-$r$ alphabet. For example, if we have a length-200 sequence of aligned DNA data for each of the $|X|$ taxa, where $|X|$ denotes the number of taxa in $|X|$, this can be summarized as a matrix $M$ on $|X|$ rows and 200 columns, where each entry of the matrix is an element from $\{A,G,C,T\}$, so $r=4$. Each of the 200 columns is then known as a \emph{character}.

Given such data,  how do we quantify the ``goodness of fit'' of the data on a given tree $T$? A classical optimality criterion for $T$ is the \emph{parsimony score} of $T$. Informally, this is the minimum number of state-changes that would necessarily be incurred along the branches of $T$ if the data observed at $X$ had evolved following the topology of the tree. If, for each character, each state is introduced at most once along the branches of the tree, we say that $T$ is a \emph{perfect phylogeny} for the data \citep{SempleSteel2003}. If such a tree $T$ exists, we say that the data permits a perfect phylogeny. The parsimony score of each character is then equal to the number of observed states (i.e. number of distinct states in the corresponding column) minus one. Perfect phylogeny is thus the best case for phylogenetic trees constructed under the popular \emph{maximum parsimony} optimality criterion, where (motivated by Occam's Razor) trees are preferred that explain the observed data with as few evolutionary changes as possible \citep{felsenstein2004inferring}. We refer to Figure~\ref{fig:threeStateExample} for clarifying examples of perfect phylogenies.



Determining whether the input data permits a perfect phylogeny is a fundamental combinatorial problem in phylogenetics, with a long history (see \citet{lam2011generalizing,shutters2013incompatible} for  excellent overviews), and it has also attracted substantial attention from the discrete optimization community
\citep{bodlaender1992two, fernandez2001perfect, Gramm2008, lam2011generalizing,misra2011generalized}. 
The latter is due to links with the literatures on (variously) graph triangulations, parameterized complexity and Steiner Trees. For binary data ($r=2$) a classical result from Buneman from 1971 states that the data permits a perfect phylogeny if and only if every pair of characters (i.e. every pair of columns) permits a perfect phylogeny \citep{buneman1971}.  A consequence of this is that, for binary data, looking only ``locally'' at the data is sufficient to determine the presence or absence of perfect phylogeny. Is testing pairs of characters also sufficient for $r \geq 3$? In 1975 Fitch refuted this claim by showing data which does not permit a perfect phylogeny, but where every pair of characters does (Fig.~\ref{fig:threeStateExample}) \citep{fitch1975toward,fitch1977problem}.  However, later it was shown that for $r=3$ the data permits a perfect phylogeny if and only if all size-3 subsets of the characters do \citep{lam2011generalizing}.

The intriguing question thus arises: is it true that, for every number of states $r \geq 2$, there 
exists a number $f(r)$ such that $r$-state data permits a 
perfect phylogeny if and only every size-$f(r)$ subset of the
characters does? To make this more concrete: could it be true that $r$-state data, irrespective of the number of characters in the input, permits a perfect phylogeny if and only if every subset of characters of size at most, say, $r^2$ permits a perfect phylogeny? How about $2^r$ instead of $r^2$? Or $2^{2^r}?$ Or is it the case that, however large we choose this function $f(r)$, at some point a sufficiently large input will be encountered whereby focusing only on size $f(r)$ subsets will deceive us into thinking that the input permits a perfect phylogeny - when in fact it does not? A conjecture, which has thus been circulating in various forms for approximately 50 years (see \citet{HabibTo2011} for a recent treatment),  states that such a constant~$f(r)$ does indeed exist for each~$r\geq 2$. This would mean that, provided $f(r)$ is chosen to grow quickly enough, there is no danger that we will be deceived: we can always determine perfect phylogeny by restricting our attention to subsets of characters of size at most $f(r)$. Here we refer to this as the \emph{local obstructions conjecture for perfect phylogeny}. Note that $f(r)$ should depend only on $r$ and no other parameters (such as $|X|$ or the number of characters in the input). We know that $f(2)=2$ and $f(3)=3$, but what about larger $r$? If the local obstructions conjecture is true, how fast does $f(r)$ grow? 

In the absence of positive progress - it is still unknown whether $f(4)$ exists - various authors have described lower bounds on $f(r)$, if it exists. 
It is known that $f(4) \geq 5$
(if it exists) \citep{HabibTo2011} and the currently
strongest general lower bound is
given in \citet{shutters2013incompatible}, where it is shown that for $r \geq 2$,
$f(r) \geq \lfloor \frac{r}{2} \rfloor  \lceil \frac{r}{2} \rceil + 1$ (if it exists). Such results do not, however, disprove the local obstructions conjecture, since $f(r)$ might still exist but grow at least quadratically. 

Here we show emphatically that the local obstructions conjecture is \emph{false}, forming a third strike against perfect phylogeny. (The first
is the NP-hardness of the problem \citep{bodlaender1992two}, and the second excludes the existence of certain parameterized algorithms \citep{bodlaender1992two,bodlaender2000hardness}). Specifically, we show that for every even $n \geq 4$ there exists an 8-state input with~$2n$ taxa and 
$2n-4$ characters with the following property: the input does \emph{not} permit a perfect phylogeny, but all proper subsets
\emph{do} permit a perfect phylogeny. This shows that, to decide whether there exists a perfect phylogeny for character data with at least~8 states, it is not enough to check all groups of a certain number of characters. It is necessary to consider all characters simultaneously. In particular, this shows that the constant $f(8)$ cannot exist (and consequently also~$f(9),f(10),\ldots$ do not exist). We emphasize that our construction can be extended to any number of taxa, odd or even, as long as it is at least~8. It is not a transient phenomenon that disappears as the number of taxa increases.

One implication of this result is the following.  For $r=2$ the fact that $f(r)$ exists forms the basis of an efficient, fixed parameter tractable algorithm for the \emph{near-perfect phylogeny} problem \citep{sridhar2007algorithms}. (See \citet{cygan2015} for an introduction to parameterized complexity). Essentially, this problem asks: ``does there exist a tree that has a parsimony score of at most $k$ with respect to the input data?'' The algorithm leverages the insight that state-changes which occur above the perfect phylogeny lower bound must occur inside small $f(2)$-size subsets of the input. Given that $f(2)$ is a constant, there are not too many size-$f(2)$ subsets and inside such a subset there are not too many places where the state change could occur. However, our result shows that such an approach is doomed to fail for $r \geq 8$. In a similar vein, the line of attack posed in \citet{shutters2013incompatible} to establish the fixed parameter tractability of the \emph{character removal} problem (i.e. deleting a minimum number of characters to obtain a perfect phylogeny), will also fail for $r \geq 8$. This is unfortunate, since datasets certainly do arise in practice with a large number of states: for amino acids $r=20$, and non-molecular character data such as that which arises in linguistics can easily have 8 or more states. Another negative consequence of our result is the following. If we allow gaps/indels in the input, we can reduce the number of states in our construction from 8 to 4. This shows that the conjecture also fails for the practical case of aligned DNA data (without relying on any complexity assumption).


\begin{figure}
\begin{subfigure}{0.3\textwidth}
 \includegraphics[scale = 0.85]{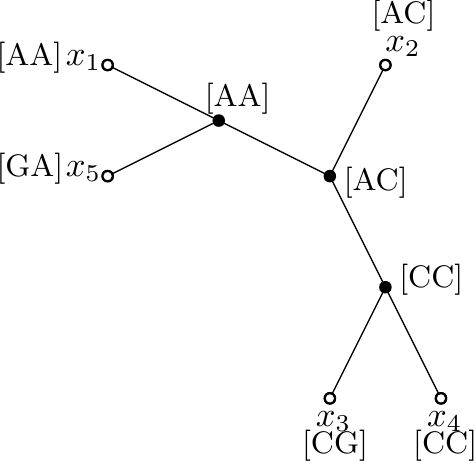}
 \caption{  $\{\chi_1,\chi_2\}$}\label{fig:ThreeStateExample1and2}    
\end{subfigure}
~
\begin{subfigure}{0.25\textwidth}
 \includegraphics[scale = 0.85]{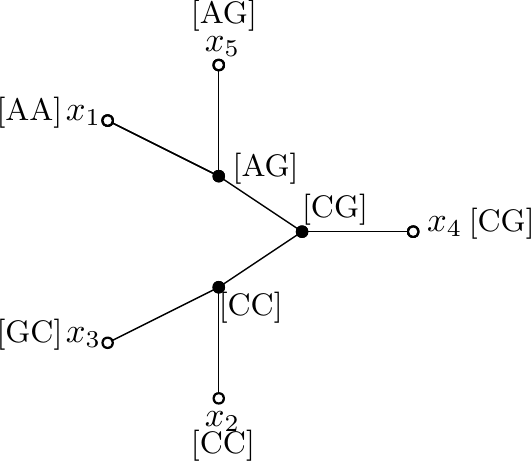}
 \caption{ $\{\chi_2,\chi_3\}$}\label{fig:ThreeStateExample2and3}     
\end{subfigure}
~
\begin{subfigure}{0.3\textwidth}
 \includegraphics[scale = 0.85]{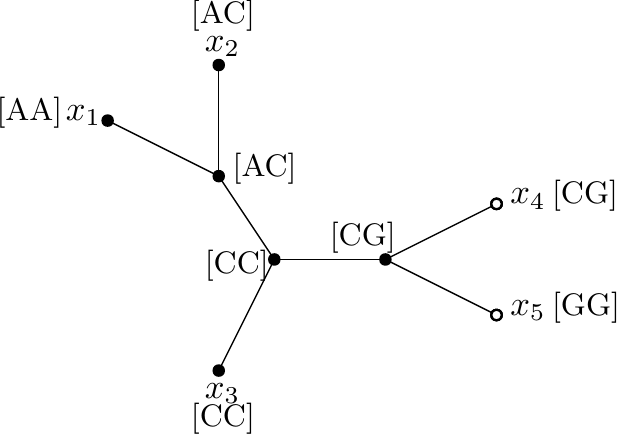}
 \caption{ $\{\chi_1,\chi_3\}$}\label{fig:ThreeStateExample1and3}    
\end{subfigure}
 \caption{The example of~\citet{fitch1975toward,fitch1977problem}, showing that~$f(3) > 2$. The five leaves~$x_1,\ldots ,x_5$ have sequences~AAA, ACC, CGC, CCG and~GAG, respectively. In our notation, the character set is $C = \{\chi_1,\chi_2,\chi_3\}$, where $\chi_1 = x_1x_2|x_3x_4|x_5$, $\chi_2 = x_1x_5|x_2x_4|x_3$, $\chi_3 = x_1|x_2x_3|x_4x_5$. For example, $\chi_1 = x_1x_2|x_3x_4|x_5$ indicates that, at the first position,~$x_1$ and~$x_2$ have the same state,~$x_3$ and~$x_4$ have the same state, and~$x_5$ has a third state. The figure shows a perfect phylogeny for each pair of characters. However, no perfect phylogeny exists for the full character set (this can easily be observed by checking that each of the three perfect phylogenies is the unique solution for its respective pair of characters).}
 \label{fig:threeStateExample}
\end{figure}

On the positive side, $f(r)$ might still exist for $r \in \{4,5,6,7\}$ (which includes the case of DNA data without gaps, i.e.,~$r=4$). Also, although our result is negative for algorithmic approaches that look only at small subsets of the input \emph{in isolation}, this is certainly \emph{not} a case of \emph{three strikes and out!} In particular, it does not exclude algorithmic approaches that analyze the input data in a more sophisticated way. For example,
the question ``does the input permit a perfect phylogeny'', although NP-hard in general \citep{bodlaender1992two}, can be answered in time
$O(2^{2r}m^2 |X|)$ using dynamic programming \citep{kannan1997fast}, which for fixed $r$ becomes $O(m^2 |X|)$. Similarly, it is still possible that fixed parameter tractable algorithms exist to solve the near-perfect phylogeny problem, but more advanced algorithmic approaches will be required. Despite the refutation of the local obstructions conjecture, perfect phylogeny will continue to play a central role in both applied and theoretical phylogenetics.

The structure of the article is as follows. We start by giving an informal description of an example of the construction for~8 taxa. After that, we give formal mathematical definitions
.
In ``Main Results'', we first describe
the most important parts of the construction of the general counter example, and explain the main ideas behind the construction. 
We then 
provide the full construction, 
and finally
prove that this gives a counterexample to the local obstructions conjecture for perfect phylogeny.


\subsection{Example for Eight Taxa}\label{sec:smallExample}

In this section, we describe our counter example for the case of~8 taxa, the smallest number of taxa for which the construction works. We describe four (6-state) characters that are incompatible, i.e., they do not permit a perfect phylogeny, while any three of the four characters do permit a perfect phylogeny. Note that this example is not a new result in itself, because it was already known that, for 6-state characters, we would need to consider at least $\lfloor \frac{6}{2} \rfloor  \lceil \frac{6}{2} \rceil + 1 = 10$ characters simultaneously~\citep{shutters2013incompatible}. Nevertheless, the example is of interest because it can be generalized to higher numbers of taxa and characters, as we will show in the remaining sections, thus proving that the local obstructions conjecture is false.

Consider eight taxa named $a_1,a_2,a_3,a_4,b_1,b_2,b_3,b_4$ and the following four characters:
\[
  \begin{array}{lr}
    \Omega_A = a_1b_1b_2|a_2a_3b_3a_4b_4 \\
    \chi_2 = a_1|b_1|a_2a_3|b_2b_3|a_4|b_4 \\
    \phi_3 = a_1a_2|b_1|b_2|a_3|a_4|b_3b_4 \\
    \Omega_B = a_1b_1a_2b_2b_3|a_3a_4b_4 
  \end{array}
\]

The names of the characters might seem odd, but they correspond to the names used in the general counter example, where they will make more sense. Also note that it actually does not matter for the problem which states taxa have. The only thing that matters is which taxa have the same state, this is indicated in the characters by separating blocks of taxa with the same state by~$|$. For example, in the first character~$\Omega_A$, taxa~$a_1,b_1$ and~$b_2$ all have the same state while $a_2,a_3,b_3,a_4$ and~$b_4$ have a different state. In~$\chi_2$, we have six states: $a_2$ and~$a_3$ have one state, $b_3$ and~$b_4$ have a second state, and the remaining four taxa all have their own unique state. The fact that we have only~6 character states is due to the small number of taxa. The general example will have~8 character states.

Figure~\ref{fig:smallCaseSubsetCompatibility} shows that any combination of three of the four characters does permit a perfect phylogeny.

\begin{figure}
\begin{center}
\begin{subfigure}{0.4\textwidth}
 \centering\includegraphics[scale = 1]{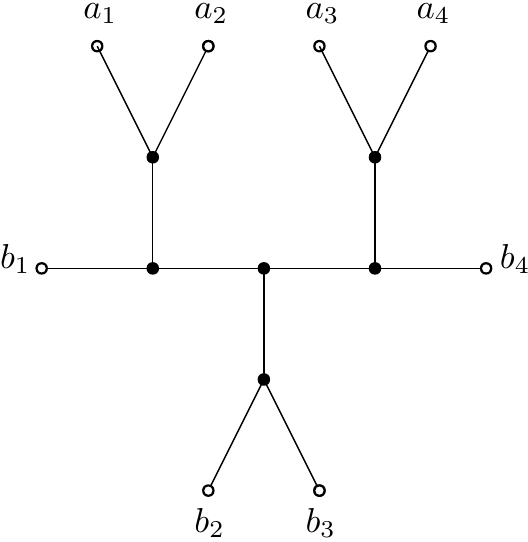}
 \caption{Tree $T_1$ displaying $\chi_2, \phi_3$ and~$\Omega_B$.}   
\end{subfigure}
\begin{subfigure}{0.4\textwidth}
 \centering\includegraphics[scale = 1]{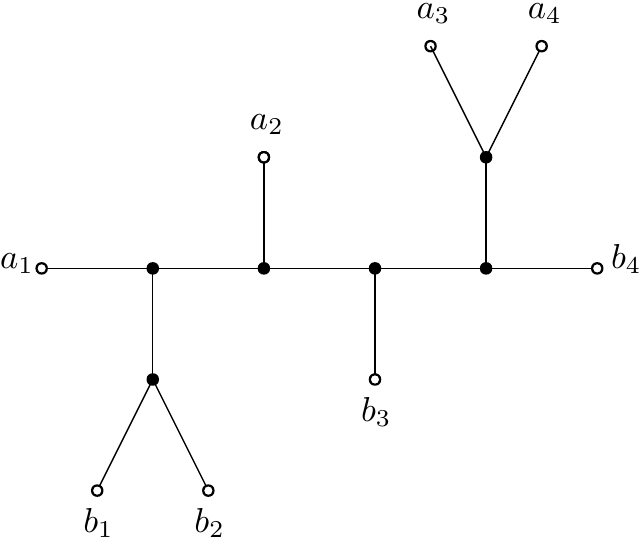}
 \caption{Tree $T_2$ displaying $\Omega_A, \phi_3$ and~$\Omega_B$.}   
\end{subfigure}

\vspace{1cm}

\begin{subfigure}{0.4\textwidth}
 \centering\includegraphics[scale = 1]{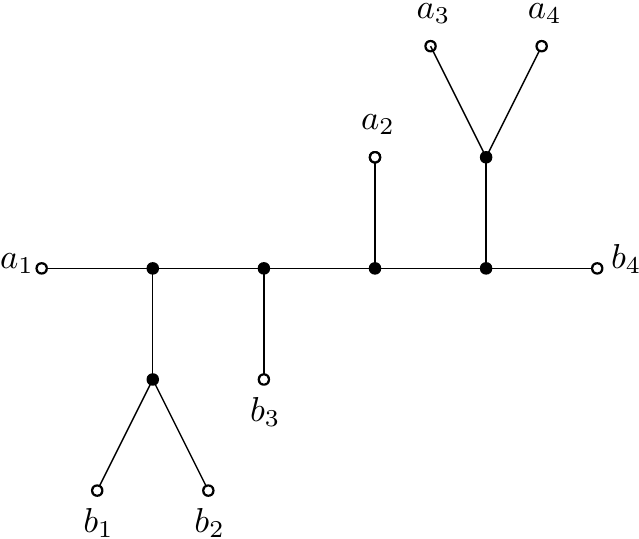}
 \caption{Tree $T_3$ displaying $\Omega_A, \chi_2$ and~$\Omega_B$.}   
\end{subfigure}
\begin{subfigure}{0.4\textwidth}
 \centering\includegraphics[scale = 1]{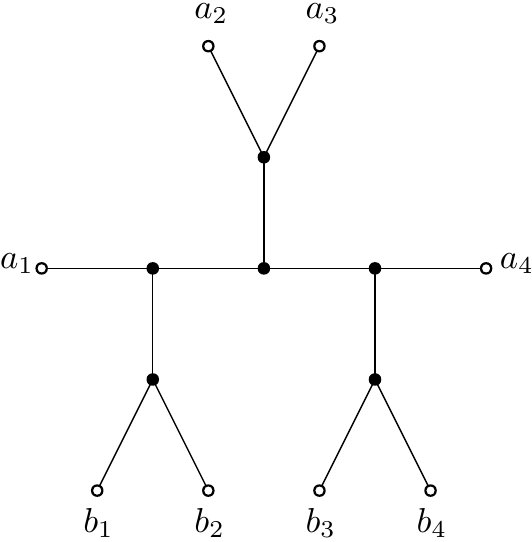}
 \caption{Tree $T_4$ displaying $\Omega_A, \chi_2$ and~$\phi_3$.}   
\end{subfigure}

\end{center}
 \caption{Four trees showing that, for eight taxa, ignoring one character of $\Omega_A, \chi_2, \phi_3$ and~$\Omega_B$ makes the remaining three characters compatible.\label{fig:smallCaseSubsetCompatibility}}
\end{figure}

We now argue that the combination of all four characters is incompatible, which is a bit more work. First we look at the characters~$\Omega_A$ and~$\Omega_B$. In character~$\Omega_A$, taxa $a_1,b_1$ and~$b_2$ all have the same state, while all other taxa have a different state. Hence, in any perfect phylogeny, there must be a branch with the taxa~$a_1,b_1$ and~$b_2$ on one side and the remaining taxa on the other side. Similarly, character~$\Omega_B$ says that there must be a branch with the taxa $a_3,a_4$ and~$b_4$ on one side and the remaining taxa on the other side. 
What the parts of the tree containing $a_1,b_1,b_2$ and $a_3,a_4,b_4$ look like is not important.
What is important is what happens in the middle part of the tree, which contains the remaining taxa~$a_2$ and~$b_3$.

Basically, characters~$\chi_2$ and~$\phi_3$ give us contradictory information about the order of taxa~$a_2$ and~$b_3$, see Figure~\ref{fig:smallCaseIncompatibility}. First look at character~$\phi_3$. Because taxa~$a_1$ and~$a_2$ have the same state, and taxa~$b_3$ and~$b_4$ have another state, we know that the path connecting~$a_1$ and~$a_2$ may not overlap with the path connecting~$b_3$ and~$b_4$. Hence~$a_2$ must be on the side of~$a_1$ and~$b_3$ on the side of~$b_4$, as indicated in 
Figure~\ref{fig:SmallCaseStructureNoChi}.
In a similar way, character~$\chi_2$ tells us exactly the opposite, i.e., that~$b_3$ is on the side of~$b_2$ (and~$a_1$) and~$a_2$ is on the side of~$a_3$ (and~$b_4$), as indicated in 
Figure~\ref{fig:SmallCaseStructureNoPhi}.
Hence, a perfect phylogeny would need to simultaneously look like 
Figure~\ref{fig:SmallCaseStructureNoChi} and like Figure~\ref{fig:SmallCaseStructureNoPhi},
which is impossible. We can therefore conclude that no perfect phylogeny exists.

\begin{figure}
\begin{center}
\begin{subfigure}{0.4\textwidth}
 \centering\includegraphics[scale = 1]{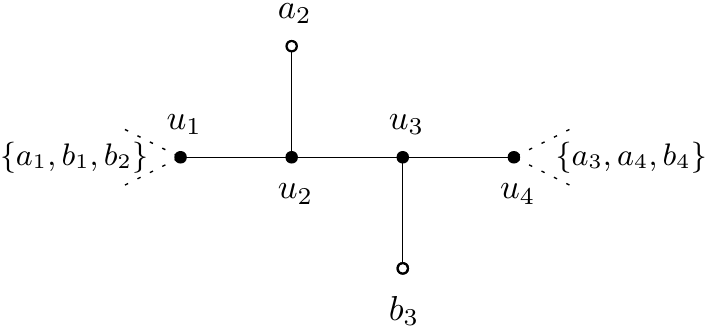}
 \caption{Structure of a perfect phylogeny implied by $\Omega_A, \phi_3, \Omega_B$} \label{fig:SmallCaseStructureNoChi}  
\end{subfigure}\hspace{1cm}
\begin{subfigure}{0.4\textwidth}
 \centering\includegraphics[scale = 1]{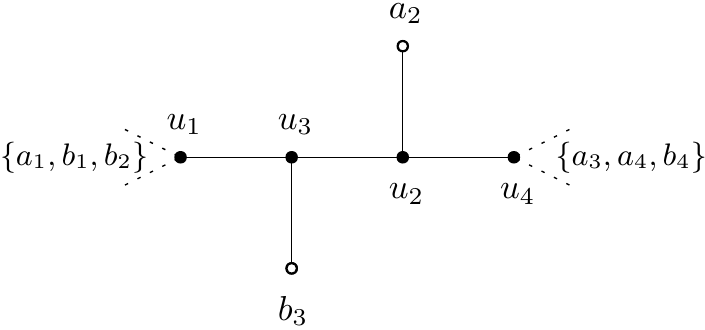}
 \caption{Structure of a perfect phylogeny implied by $\Omega_A, \chi_2, \Omega_B$}   \label{fig:SmallCaseStructureNoPhi}
\end{subfigure}
\end{center}
 \caption{Illustration of the proof that, for eight taxa, characters $\Omega_A, \chi_2, \phi_3$ and~$\Omega_B$ are incompatible.\label{fig:smallCaseIncompatibility}}
\end{figure}

In the remaining sections, we show how to generalize this example to more taxa, thereby also increasing the number of characters. We note that the proofs for the general case will be more involved.

\subsection{Mathematical Definitions}\label{sec:preliminaries}

Let $X$ be a set of labels.
For any positive integer $r$, an \emph{$r$-state character on $X$} is a partition $\chi = S_1|S_2|\dots|S_{r'}$, where $X$ is the union of $S_1, \dots, S_{r'}$ and $r' \leq r$.
We refer to the sets $S_1, \dots, S_{r'}$ as \emph{states}.
For the sake of brevity, in this context we will sometimes write $x_1\dots x_t$ as shorthand for a set $\{x_1, \dots, x_t\}$. Thus for example, if $X = \{x_1, \dots ,x_n\}$ then $\chi = x_1|x_2|x_3x_4|\{x_i: i \geq 5\}$ is a character on $X$.
(Note that some states may be empty; in such cases we may treat these states as non-existent. Thus for example if $S_i = \emptyset$ then $ S_1|S_2|\dots|S_{r'}$ is equivalent to $S_1|S_2|\dots |S_{i-1}|S_{i+1}|\dots|S_{r'}$.)

A \emph{tree $T$ on $X$} is an unrooted tree with leaves bijectively labelled with the elements of $X$.
Given a subset $S \subset X$, let $T[S]$ denote the minimal subtree of $T$ whose vertices contain $S$.
We note that degree-$2$ vertices are usually not allowed in phylogenetic trees; however our definition of $T[S]$ allows for degree-$2$ vertices, as this makes certain proofs simpler and does not affect the results.

For any positive integer $n$, $[n]$ denotes the set $\{1, \dots, n\}$.
We say $T$ \emph{displays} a character $\chi = S_1 | \dots | S_{r'}$ on $X$ if there exists a partition $V_1 |\dots | V_{r'}$ of the vertices of $T$, such that the subtree of $T$ induced by $V_i$ is connected and $V_i \cap X = S_i$ for each $i \in [r']$.
Equivalently, $T$ displays $\chi$ if the subtrees $T[S_i]$ and $T[S_j]$ are vertex-disjoint for $i \neq j$.
We say $T$ is \emph{compatible} with a set $C$ of characters (or equivalently, $C$ is \emph{compatible} with $T$)  if $T$ displays $\chi$ for each $\chi \in C$.
If this is the case, we also say that $T$ is a \emph{perfect phylogeny for $C$}.
We say a set $C$ of characters is \emph{compatible} if there exists a perfect phylogeny for $C$.

In this paper, we show that the following conjecture is false:

\begin{conjecture}\label{con:perfectPhylogenyCharacterization}
 For each positive integer $r$, there exists an integer $f(r)$ such that for any finite set $X$ and any set $C$ of $r$-state characters on $X$,
 $C$ is compatible if and only if every subset of at most $f(r)$ characters in $C$ is compatible.
\end{conjecture}

\setcounter{section}{2}
\setcounter{theorem}{0}
\section{Main Results}
\subsection{Counterexample: Main Concepts}\label{sec:mainConcepts}

In this section, we outline the main concepts and ideas used in the construction of our counterexample to Conjecture~\ref{con:perfectPhylogenyCharacterization}. 
We also define the label set $X$ and two trees on $X$ that will be used to show that most subsets of characters are compatible.

In what follows, let $n$ be any positive even integer.

\begin{definition}
 Given a positive even integer $n$,
 let $X = \{a_1, \dots a_n, b_1, \dots, b_n\}$.
 For any $i \in [n]$, let $X_{\leq i} = \{a_j,b_j: 1 \leq j \leq i\}$, and $X_{\geq i} = \{a_j,b_j: m \geq j \geq i\}$.
\end{definition}

We now define two trees $A$ and $B$ on $X$.
These trees appear quite similar on a large scale - they are both lobsters (trees in which every vertex is of distance at most $2$ from a central path), with leaves of smaller index closer to one end of the central path than leaves of larger index.
However, on a local scale they appear quite different - for example, each $x\in X$ has a different sibling in $A$ than in $B$.

Informally, $A$ consists of a number of cherries that are attached as pendant subtrees to a central path. 
The endpoints of the path are $a_1$ and $a_n$. Starting at $a_1$ and walking along the path, the first cherry attached is $(b_1,b_2)$, then $(a_2,a_3)$, then $(b_3,b_4)$, and so on.
The definition of tree $B$ is similar to $A$, but with the roles of the $a$ and $b$ leaves reversed.
(Fig.~\ref{fig:BasicLobsters}.)

We give a more formal definition below.

\begin{definition}
The tree $A$ on $X$ is defined as follows:
$A$ has leaves $a_1, \dots, a_n$, $b_1, \dots, b_n$ and internal nodes $u_1, \dots, u_{n-1}$, $v_1, \dots, v_{n-1}$.
$A$ contains a central path $a_1, u_1$, $u_2,\dots, u_{n-1}$, $a_n$.
For each $i \in [n-1]$, there is an edge $u_iv_i$.
For odd $i \in [n-1]$, the vertex $v_i$ is adjacent to leaves $b_i$ and $b_{i+1}$.
For even $i \in[n-2]$, the vertex $v_i$ is adjacent to leaves $a_i$ and $a_{i+1}$.

The tree $B$ on $X$ is defined as follows:
$B$ has leaves $a_1, \dots, a_n$, $b_1, \dots b_n$ and internal nodes $u_1, \dots, u_{n-1}$, $v_1, \dots, v_{n-1}$.
$B$ contains a central path $b_1, u_1$, $u_2,\dots, u_{n-1}$, $b_n$.
For each $i \in [n-1]$, there is an edge $u_iv_i$.
For odd $i \in [n-1]$, the vertex $v_i$ is adjacent to leaves $a_i$ and $a_{i+1}$.
For even $i \in[n-2]$, the vertex $v_i$ is adjacent to leaves $b_i$ and $b_{i+1}$.
\end{definition}

\begin{figure}
\begin{subfigure}{0.3\textwidth}
 \centering\includegraphics[scale = 0.85]{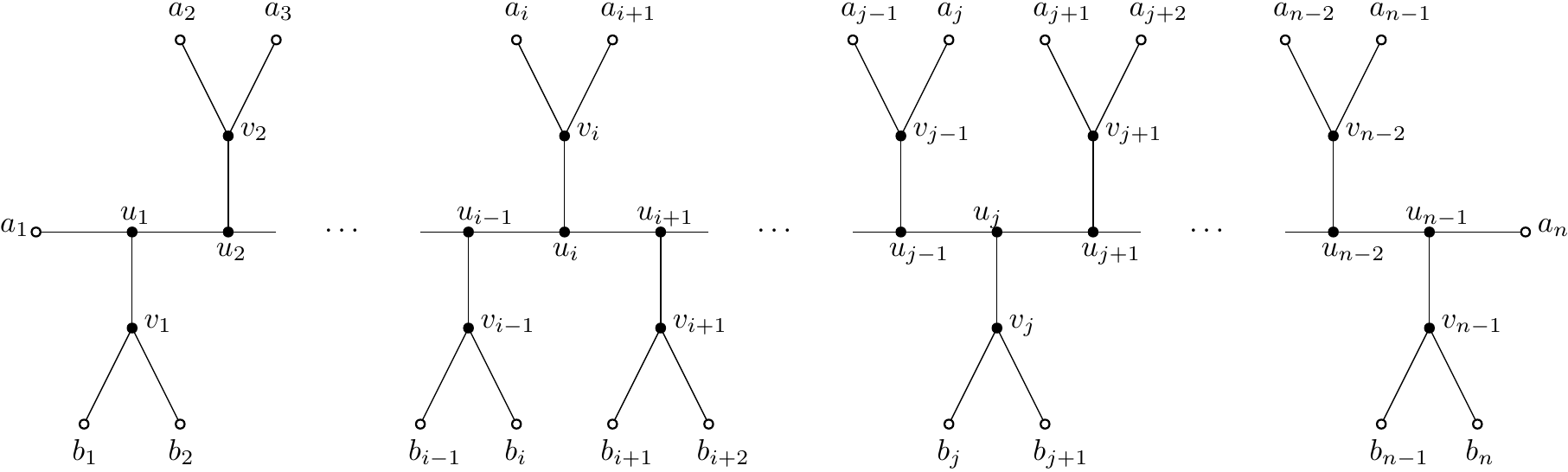}
 \caption{ Lobster $A$}    
\end{subfigure}

\vspace{1cm}

\begin{subfigure}{0.3\textwidth}
 \centering\includegraphics[scale = 0.85]{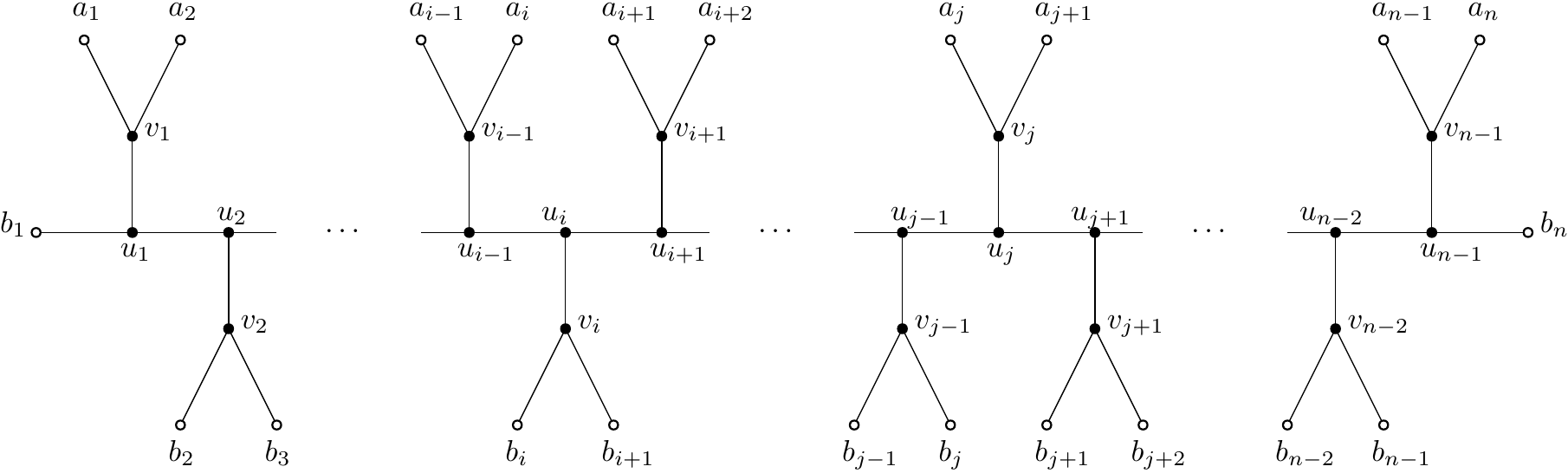}
 \caption{Lobster $B$}    
\end{subfigure}
 \caption{The lobsters $A$ and $B$. The middle part of each figure shows some of the vertices near $u_i$, for $i$ even, and near $u_j$, for $j > i$ and $j$ odd.}\label{fig:BasicLobsters}
\end{figure}

We next describe a set of characters $\chi_i$ for each 
$i$ such that $2 \leq i \leq n-2$
(we note that this is not the full set of characters that will be used in the complete example).
Informally, each character $\chi_i$ can be thought of as caring about a small local part of the tree.
It roughly enforces that if one segment of the tree looks like $A$, then  so does the next segment along.

For each $2 \leq i \leq n-2$, define

$$\chi_i =  X_{\leq i-2}|a_{i-1}|b_{i-1}|a_ia_{i+1}|b_i b_{i+1}|a_{i+2}|b_{i+2}|X_{\geq i+3}$$

(Note that for $i=2$ the set $X_{\leq i-2}$ is empty; thus $\chi_2$ could be equivalently written as $a_{1}|b_{1}|a_2a_{3}|b_2 b_{3}|a_{4}|b_{4}|X_{\geq 5}$. Similarly, for $i = n-2$ the set $X_{\geq i+3}$ is empty and so $\chi_{n-2}$ can be written as  $X_{\leq n-4}|a_{n-3}|b_{n-3}|a_{n-2}a_{n-1}|b_{n-2} b_{n-1}|a_{n}|b_{n}$.)

Observe that both $A$ and $B$ display $\chi_i$ for each $2 \leq i \leq n-2$, but the structure of the subtrees involved is quite different between the two.
In particular, assuming $i$ is even, in $A$ the path from $a_{i}$ to $a_{i+1}$ has length $2$, whereas in $B$ the same path has length $6$, and similarly in $A$ the path from $b_{i}$ to $b_{i+1}$ has length $6$, whereas in $B$ it has length $2$.
(See Figure~\ref{fig:ChiEmbeddings} for an example when $i$ is even.)

\begin{figure}
\begin{center}
\begin{subfigure}{0.3\textwidth}
 \includegraphics[scale = 0.85]{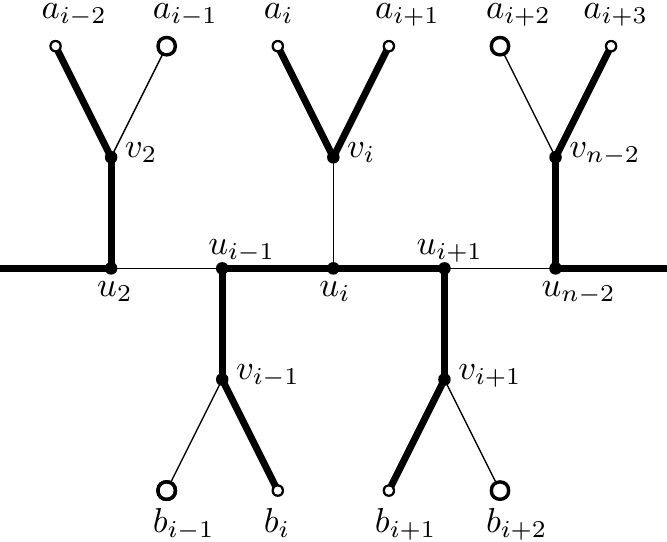}
 \caption{ Lobster $A$}    
\end{subfigure}

\vspace{1cm}

\begin{subfigure}{0.3\textwidth}
 \includegraphics[scale = 0.85]{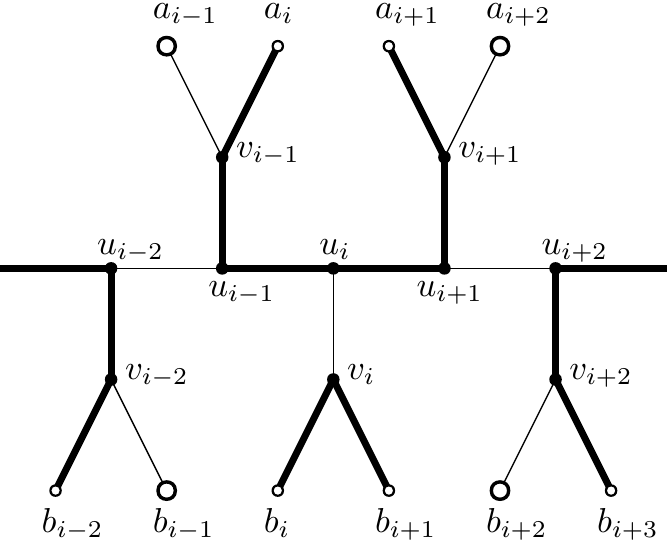}
 \caption{Lobster $B$}    
\end{subfigure}
\end{center}
 \caption{How $A$ and $B$ each display the character $\chi_i$, for $i$ even.}\label{fig:ChiEmbeddings}
 \label{fig:caterpillar}
\end{figure}

With the addition of further characters (to be described in the next section), we will be able to enforce that $A$ and $B$ are in fact the only trees compatible with all those characters.
In order to make the full set of characters incompatible, we will add two more characters $\Omega_A$ and $\Omega_B$, defined as follows:

$$\Omega_A = a_1b_1b_2|\{a_2\} \cup X_{\geq 3}$$

$$\Omega_B = X_{\leq n-2}\cup \{b_{n-1}\}|a_{n-1}a_{n}b_n$$

Observe that $\Omega_A$ is displayed by $A$ but not by $B$, while $\Omega_B$ is displayed by $B$ but not by $A$.

We will claim that 
every strict subset of this set of characters is compatible.
In order to show this, we will prove that for each integer $i$ between $2$ and $n-2$, there is a tree displaying all characters except $\chi_i$.
The intuition here is as follows: 
$\chi_i$ enforces something about the local structure of a perfect phylogeny; in particular it is the only character in the constructed set requiring 
that the path from $a_i$ to $a_{i+1}$ and the path from $b_i$ to $b_{i+1}$ are vertex-disjoint.
Removing $\chi_i$ allows us to consider $X$ as being made of two parts: $X_{\leq i}$ and $X_{\geq i+1}$. We can construct a tree which is isomorphic to $A$ when restricted to  $X_{\leq i}$ , and isomorphic to $B$ when restricted to $X_{\geq i+1}$.
Such a tree is denoted $A_iB$, and is defined below
(Fig.~\ref{fig:CrossLobsters}).

\begin{definition}\label{def:ACrossB}
For $2 \leq i \leq n-2$, 
the tree $A_iB$ on $X$ is defined as follows:
$A_iB$ has leaves $a_1, \dots, a_n, b_1, \dots, b_n$ and internal nodes $u_1, \dots, u_{i-1}$, $u_{i+1}, \dots, u_{n-1}$, $v_1, \dots, v_{i-1}$, $v_{i+1}, \dots,  v_{n-1}$, $u_A$, $u_B$
(note that $A_iB$ does not have vertices $u_i$ or $v_i$ but instead has $u_A$ and $u_B$).
 $A_iB$ contains a central path $a_1$,$u_1$,$u_2,\dots, u_{i-1}$,$u_A$, $u_B$,$u_{i+1}, \dots, u_{n-1}$,$b_n$.
For each $j \in [n-1] \setminus \{i\}$, there is an edge $u_jv_j$.
If $i$ is even then $u_A$ is adjacent to $a_i$ and $u_B$ is adjacent to $b_{i+1}$.
On the other hand if $i$ is odd then $u_A$ is adjacent to $b_i$ and $u_B$ is adjacent to $a_{i+1}$.
For $j < i$, the vertex $v_j$ is adjacent to  $b_j$ and $b_{j+1}$ if $j$ is odd, and adjacent to  $a_j$ and $a_{j+1}$ if $j$ is even.
For $j > i$, $v_j$ is adjacent to  $a_j$ and $a_{j+1}$ if $j$ is odd, and adjacent to  $b_j$ and $b_{j+1}$ if $j$ is even.
\end{definition}

\begin{figure}
\begin{center}
\begin{subfigure}{\textwidth}
 \centering\includegraphics[scale = 0.85]{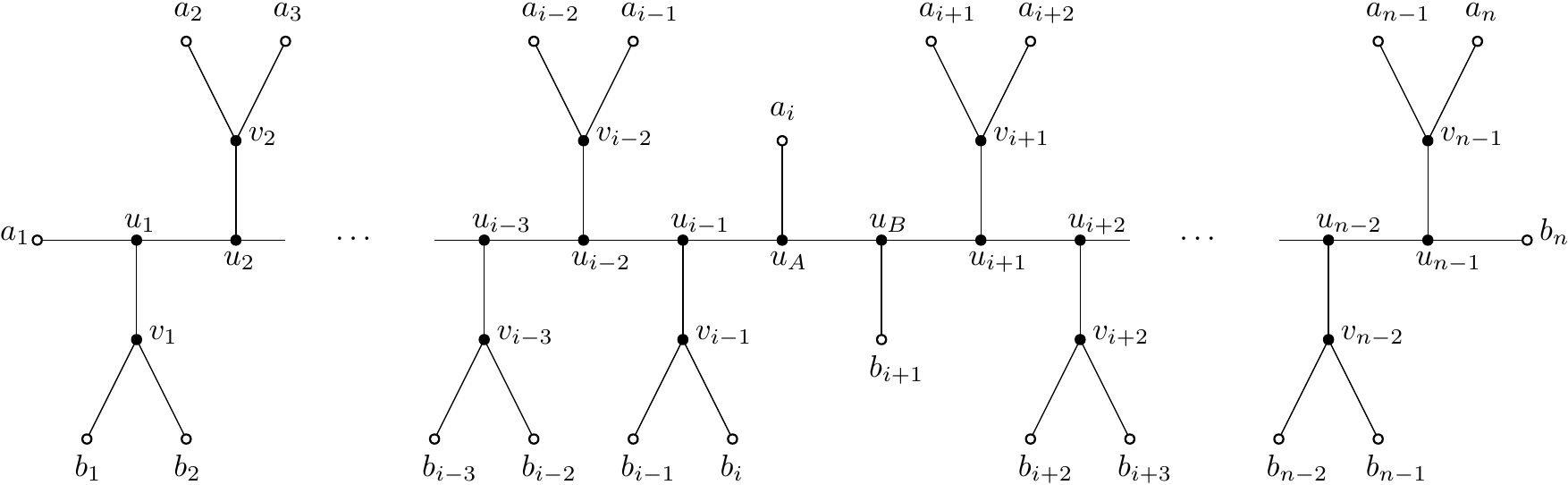}
 \caption{ Lobster $A_iB$ for $i$ even}    
\end{subfigure}

\vspace{1cm}

\begin{subfigure}{\textwidth}
 \centering\includegraphics[scale = 0.85]{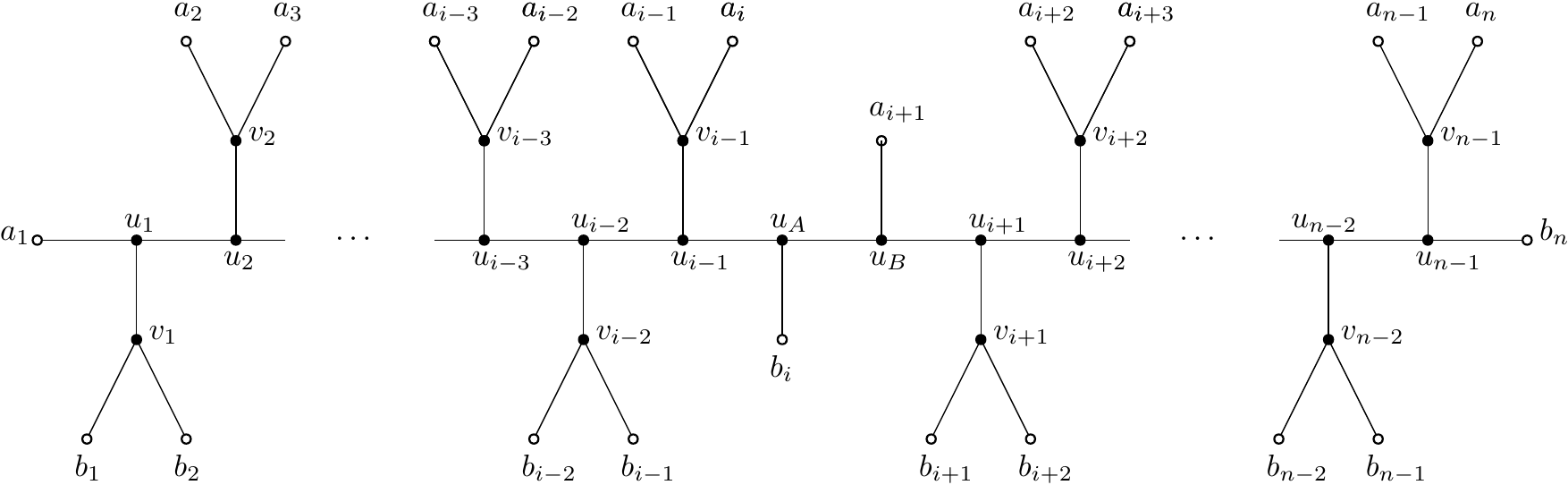}
 \caption{Lobster $A_iB$ for $i$ odd}    
\end{subfigure}
\end{center}
 \caption{The lobster $A_iB$, for the cases when $i$ is even and $i$ is odd.}\label{fig:CrossLobsters}
\end{figure}

%
%
%

Observe that $A_iB$ does not display $\chi_i$, but it does display $\Omega_A$ and $\chi_j$ for each $j < i$ (by a similar argument to how $A$ displays those characters), and it does display $\Omega_B$ and $\chi_j$ for each $j > i$ (by a similar argument to how $B$ displays those characters).

It follows that any
strict subset of characters in the set
is compatible 
(since any subset missing $\Omega_A$ is compatible with $B$, any subset missing $\Omega_B$ is compatible with $A$, and any subset missing $\chi_i$ for some $2 \leq i \leq n-2$ is compatible with $A_iB$).

In the next section, we make the concepts described above more formal. 
The main work will be to define additional characters 
(used to enforce that 
any tree compatible with all characters except $\Omega_B$ must have a similar structure to $A$),
and then to prove formally that the observations outlined above 
(that the full set of characters is incompatible, and that it becomes compatible if 
any character
is removed)
hold when the new characters are considered.

\subsection{Full Counterexample}\label{sec:fullCounterexample}

We now describe the full set $C$ of $8$-state characters on $X$.
$C$ will be a set that is incompatible, but such that every 
strict subset of $C$ is compatible.
In what follows we assume that $n$ is a positive even integer, and that $n \geq 6$ (as we already gave a counterexample with $n = 4$ in the introduction). 

(In order to avoid tedious repetition of definitions, for some values of $j$ the characters below may be described as containing elements $a_h$ or $b_h$ for $h \notin [n]$. Such elements
should be treated as non-existent, as they are not in $X$. Note that certain states of  some characters will be empty as a result.)

\begin{definition}\label{def:characters}
 For each $2 \leq j \leq n-2$, define the following character on $X$:
 
$$\chi_j =  X_{\leq j-2}|a_{j-1}|b_{j-1}|a_ja_{j+1}|b_j b_{j+1}|a_{j+2}|b_{j+2}|X_{\geq j+3}$$

For $3 \leq j \leq n-1$, define the following characters:


\[\phi_j = \left\{
  \begin{array}{lr}
      X_{\leq j-3} \cup b_{j-2}b_{j-1}|a_{j-2}|a_{j-1}|b_{j}|b_{j+1}| a_{j}a_{j+1}\cup X_{\geq j+2}   \mbox{ if $j$ is even}\\
    X_{\leq j-3} \cup a_{j-2}a_{j-1}|b_{j-2}|b_{j-1}|a_{j}|a_{j+1}| b_{j}b_{j+1}\cup X_{\geq j+2}  \mbox{ if $j$ is odd}
  \end{array}
\right.
\]

Finally define the two characters:

$$\Omega_A = a_1b_1b_2|\{a_2\} \cup X_{\geq 3}$$

$$\Omega_B = X_{\leq n-2}\cup \{b_{n-1}\}|a_{n-1}a_{n}b_n$$

Let $C$ be the set of all 2n-4 characters described above.
\end{definition}

Observe that the construction of $C$ is the same as in 
the section ``Counterexample: Main Concepts''.
with the addition of characters $\phi_j$ for $3 \leq j \leq n-1$.
In the remainder of this section, we show that every strict subset of $C$ is compatible and that $C$ itself is incompatible.
We begin by proving formally that the lobster $A$ displays every character in $C$ except for $\Omega_B$.

\subsubsection{Compatibility of $C\setminus \{\Omega_B\}$}\label{sec:compatibleNoOmegaB}

Before continuing, we note that if a state $S$ of some character $\chi$ consists of a single element of $X$, then for any tree $T$ on $X$, the subtree $T[S]$  
 is automatically vertex-disjoint from $T[S']$ for any other state $S'$ of $\chi$.
 This is because $T[S]$ consists only of a single leaf in $T$, and as $S$ and $S'$ are disjoint, $T[S']$ does not contain that leaf.
 Therefore when showing that a tree displays a particular character, we may focus on the states of size at least $2$ in that character.

\begin{lemma}\label{lem:ADisplaysOmegaA}
 Lobster $A$ displays $\Omega_A$.
\end{lemma}
\begin{proof}
 Observe that cutting the edge $u_1u_2$ separates $A$ into two trees, one with leaves $a_1,b_1,b_2$ and one with leaf set $\{a_2\} \cup X_{\geq 3}$. It follows that the subtrees of $AB$ spanning these two sets are vertex-disjoint, and so $A$ displays $\Omega_A$.
\end{proof}

\begin{lemma}\label{lem:ADisplaysChi}
 For each $2 \leq j \leq n-2$, lobster $A$ displays $\chi_j$.
\end{lemma}
\begin{proof}
The non-singleton states of $\chi_j$ are  $X_{\leq j-2}$, $\{a_j,a_{j+1}\}$, $\{b_j, b_{j+1}\}$, and $X_{\geq j+3}$.
 Cutting the edge  $u_{j-2}u_{j-1}$ separates $X_{\leq j-2}$ from the other non-singleton states.
 Similarly, cutting the edge $u_{j+1}u_{j+2}$ separates $X_{\leq j+3}$ from the other non-singleton states.
 It remains to show that the trees $A[\{a_j,a_{j+1}\}]$, $A[\{b_j, b_{j+1}\}]$ are vertex disjoint.
 This can be seen by cutting the edge $u_jv_j$ (as $v_j$ is adjacent either to the leaves $a_j$ and $a_{j+1}$, or to the leaves $b_j$ and $b_{j+1}$, depending on whether $j$ is even or odd).
\end{proof}

\begin{lemma}\label{lem:ADisplaysPhi}
 For each $3 \leq j \leq n-1$, lobster $A$ displays $\phi_j$.
\end{lemma}
\begin{proof}
   The character $\phi_j$ has two non-singleton states.
   If $j$ is even, then the non-singleton states are  $X_{\leq j-3} \cup \{b_{j-2},b_{j-1}\}$ and $\{a_{j},a_{j+1}\}\cup X_{\geq j+2}$.
   Note that in this case $b_{j-2}$ is adjacent to $v_{j-3}$, $b_{j-1}$ is adjacent to $v_{j-1}$,  and $a_j$ and $a_{j+1}$ are both  adjacent to $v_j$.
   It follows that cutting the edge $u_{j-1}u_j$ (which separates $v_j$ from $v_{j-1}$ and $v_{j-3}$) will separate the two non-singleton states from each other.
   
   If $j$ is odd, then the non-singleton states are  $X_{\leq j-3} \cup \{a_{j-2},a_{j-1}\}$ and $\{b_{j},b_{j+1}\}\cup X_{\geq j+2}$.
   In this case, $a_{j-2}$ is adjacent to $v_{j-3}$ (unless $j=3$, in which case $a_{j-2} = a_1$ is adjacent to $u_1 = u_{j-2}$), $a_{j-1}$ is adjacent to $v_{j-1}$, and  $b_j$ and $b_{j+1}$ are adjacent to $v_j$.
   Thus, we again have that cutting the edge $u_{j-1}u_j$ will separate the two non-singleton states from each other.
\end{proof}

The next lemma follows from Lemmas~\ref{lem:ADisplaysOmegaA},~\ref{lem:ADisplaysChi} and~\ref{lem:ADisplaysPhi}.

\begin{lemma}\label{lem:CminusOmegaACompatible}
 Lobster $A$ is compatible with $C \setminus \{\Omega_B\}$.
\end{lemma}

\subsubsection{Compatibility of $C\setminus \{\Omega_A\}$}\label{sec:compatibleNoOmegaA}
We next prove formally that the lobster $B$ displays every character in $C$ except for $\Omega_A$.
The proofs here are very similar to those for $A$.

\begin{lemma}\label{lem:BDisplaysOmegaB}
 Lobster $B$ displays $\Omega_B$.
\end{lemma}
\begin{proof}
 Observe that cutting the edge $u_{n-2}u_{n-1}$ separates $B$ into two trees, one with leaves $a_{n-1},a_n,b_n$ and one with leaf set $X_{\leq n-2} \cup \{b_{n-1}\}$. It follows that the subtrees of $B$ spanning these two sets are vertex-disjoint, and so $B$ displays $\Omega_B$.
\end{proof}

\begin{lemma}\label{lem:BDisplaysChi}
 For each $2 \leq j \leq n-2$, lobster $B$ displays $\chi_j$.
\end{lemma}
\begin{proof}
The non-singleton states of $\chi_j$ are  $X_{\leq j-2}$, $\{a_j,a_{j+1}\}$, $\{b_j, b_{j+1}\}$, and $X_{\geq j+3}$.
 Cutting the edge  $u_{j-2}u_{j-1}$ separates $X_{\leq j-2}$ from the other non-singleton states.
 Similarly, cutting the edge $u_{j+1}u_{j+2}$ separates $X_{\leq j+3}$ from the other non-singleton states.
 It remains to show that the trees $B[\{a_j,a_{j+1}\}]$, $B[\{b_j, b_{j+1}\}]$ are vertex disjoint.
 This can be seen by cutting the edge $u_jv_j$ (as $v_j$ is adjacent either to the leaves $a_j$ and $a_{j+1}$, or to the leaves $b_j$ and $b_{j+1}$, depending on whether $j$ is even or odd).
\end{proof}

\begin{lemma}\label{lem:BDisplaysPhi}
 For each $3 \leq j \leq n-1$, lobster $B$ displays $\phi_j$.
\end{lemma}
\begin{proof}
   The character $\phi_j$ has two non-singleton states.
   If $j$ is even, then the non-singleton states are  $X_{\leq j-3} \cup \{b_{j-2},b_{j-1}\}$ and $\{a_{j},a_{j+1}\}\cup X_{\geq j+2}$.
   Note that in this case $b_{j-2}$ and $b_{j-1}$ are adjacent to $v_{j-2}$,   $a_j$ is adjacent to $v_{j-1}$, and $a_{j+1}$ is  adjacent to $v_{j+1}$.
   It follows that cutting the edge $u_{j-2}u_{j-1}$ (which separates $v_{j-2}$ from $v_{j-1}$ and $v_{j+1}$) will separate the two non-singleton states from each other.
   
   If $j$ is odd, then the non-singleton states are  $X_{\leq j-3} \cup \{a_{j-2},a_{j-1}\}$ and $\{b_{j},b_{j+1}\}\cup X_{\geq j+2}$.
   In this case, $a_{j-2}$ and $a_{j-1}$ are adjacent to $v_{j-2}$, $b_j$ is adjacent to $v_{j-1}$, and $b_{j+1}$ is adjacent to $v_{j+1}$ (unless $j = n-1$, in which case $b_{j+1} = b_n$ is adjacent to $u_{n-1} = u_j$).
   Thus, we again have that cutting the edge $u_{j-2}u_{j-1}$ will separate the two non-singleton states from each other.
\end{proof}

The next lemma follows from Lemmas~\ref{lem:BDisplaysOmegaB},~\ref{lem:BDisplaysChi} and~\ref{lem:BDisplaysPhi}.

\begin{lemma}\label{lem:CminusOmegaBCompatible}
 Lobster $B$ is compatible with $C \setminus \{\Omega_A\}$.
\end{lemma}

\subsubsection{Compatibility of $C\setminus \{\chi_i\}$ for each $2 \leq i \leq n-2$}\label{sec:compatibleNoChi}

We now show that for any $2 \leq i \leq n-2$, the set $C \setminus \{\chi_i\}$ is compatible.
Recall the definition of Lobster $A_iB$ (Definition~\ref{def:ACrossB} and Fig.~\ref{fig:CrossLobsters}).  
We will show that $A_iB$ displays every character in $C$ except for $\chi_i$.

Recall that $A_iB$ restricted to  $X_{\leq i}$ is isomorphic to $A[X_{\leq i}]$, while $A_iB$ restricted to $X_{\geq i+1}$ is isomorphic to $B[X_{\geq i+1}]$.

\begin{lemma}\label{lem:AcrossBDisplaysOmegaAB}
 For any $2 \leq i \leq n-2$, lobster $A_iB$ displays $\Omega_A$ and  $\Omega_B$.
\end{lemma}
 \begin{proof}
 To see that $A_iB$ displays $\Omega_A$, observe that cutting the edge $u_1u_2$ (or $u_1u_A$ if $i=2$) separates $A_iB$ into two trees, one with leaves $a_1,b_1,b_2$ and one with leaf set $\{a_2\} \cup X_{\geq 3}$. It follows that the subtrees of $A_iB$ spanning these two sets are vertex-disjoint, and so $A_iB$ displays $\Omega_A$.
 Similarly, to see that $A_iB$ displays $\Omega_B$, observe that cutting the edge $u_{n-2}u_{n-1}$ (or $u_Bu_{n-1}$ if $i = n-2$) separates $A_iB$ into two trees, with leaf sets $X_{\leq n-2}\cup \{b_{n-1}\}$ and $\{a_{n-1}a_nb_n\}$ respectively.
 \end{proof}

\begin{lemma}\label{lem:AcrossBDisplaysMostChi}
For any $2\leq i, j \leq n-2$ such that $i \neq j$,
lobster $A_iB$ displays $\chi_j$.
\end{lemma}

\begin{proof}
The non-singleton states of $\chi_j$ are  $X_{\leq j-2}$, $\{a_j, a_{j+1}\}$, $\{b_j, b_{j+1}\}$ and $X_{\geq j+3}$.
 Cutting the edge  $u_{j-2}u_{j-1}$  ($u_{j-2}u_A$ if $j = i+1$, $u_Bu_{j-1}$ if $j= i +2$) separates $X_{\leq j-2}$ from the other non-singleton states.
 Similarly, cutting the edge $u_{j+1}u_{j+2}$ ($u_{j+1}u_A$ if $j = i -2$, $u_Bu_{j+2}$ if $j= i -1$) separates $X_{\leq j+3}$ from the other non-singleton states.
 It remains to show that the trees $A_iB[\{a_ja_{j+1}\}]$, $A_iB[\{b_j b_{j+1}\}]$ are vertex disjoint.
 This can be seen by cutting the edge $u_jv_j$ (as $v_j$ is adjacent either to the leaves $a_i$ and $a_{i+1}$, or to the leaves $b_i$ and $b_{i+1}$).
\end{proof}

\begin{lemma}\label{lem:AcrossBDisplaysPhi}
 For any $2 \leq i \leq n-2$ and for each $3 \leq j \leq n-1$, lobster $A_iB$ displays $\phi_j$.
\end{lemma}
\begin{proof}
  The character $\phi_j$ has two non-singleton states; 
  these are either $X_{\leq j-3} \cup \{b_{j-2},b_{j-1}\}$ and $\{a_{j},a_{j+1}\}\cup X_{\geq j+2}$ (if $j$ is even)  or $X_{\leq j-3} \cup \{a_{j-2},a_{j-1}\}$ and $\{b_{j},b_{j+1}\}\cup X_{\geq j+2}$ (if $j$ is odd).
  
  We first consider the case when $j \notin \{i,i+1,i+2\}$.
  In this case there are four possibilities to consider: 
  \begin{itemize}
  \item  If $j$ is even and $j < i$,
  then cutting the edge $u_{j-1}u_j$ separates $X_{\leq j-3} \cup \{b_{j-2},b_{j-1}\}$ from $\{a_{j},a_{j+1}\}\cup X_{\geq j+2}$. 
  \item If $j$ is even and $j > i+2$, then cutting the edge $u_{j-2}u_{j-1}$ separates $X_{\leq j-3} \cup \{b_{j-2},b_{j-1}\}$ from $\{a_{j},a_{j+1}\}\cup X_{\geq j+2}$.
  \item   If $j$ is odd and $j < i$, then cutting the edge $u_{j-1}u_j$ separates $X_{\leq j-3} \cup \{a_{j-2},a_{j-1}\}$ from $\{b_{j},b_{j+1}\}\cup X_{\geq j+2}$.
  \item  If $j$ is odd and $j > i+2$, then cutting the edge $u_{j-2}u_{j-1}$ separates $X_{\leq j-3} \cup \{a_{j-2},a_{j-1}\}$ from $\{b_{j},b_{j+1}\}\cup X_{\geq j+2}$.
  \end{itemize}

  We now consider the case when $j \in \{i,i+1,i+2\}$, and suppose first that $i$ is even.
  \begin{itemize}
  \item   If $j = i$, then cutting $u_{i-1}u_A$ separates $X_{\leq j-3} \cup \{b_{j-2},b_{j-1}\}$ from $\{a_{j},a_{j+1}\}\cup X_{\geq j+2}$.
  \item  If $j = i+1$ then cutting the edge $u_Au_B$ separates $X_{\leq i-2} \cup \{a_{i-1},a_{i}\} = X_{\leq j-3} \cup \{a_{j-2},a_{j-1}\}$ from $\{b_{i+1},b_{i+2}\}\cup X_{\geq j+3} = \{b_{j},b_{j+1}\}\cup X_{\geq j+2}$.
  \item If $j = i+2$, then cutting $u_Bu_{i+1}$ separates $X_{\leq i-1} \cup \{b_{i},b_{i+1}\} = X_{\leq j-3} \cup \{b_{j-2},b_{j-1}\}$ from $\{a_{i+2},a_{i+3}\}\cup X_{\geq i+4} = \{a_{j},a_{j+1}\}\cup X_{\geq j+2}$.
  \end{itemize}

  Finally, consider the case when $j \in \{i,i+1,i+2\}$, and $i$ is odd.
  \begin{itemize}
  \item If $j = i$, then $u_{i-1}u_A$ separates $X_{\leq j-3} \cup \{a_{j-2},a_{j-1}\}$ from $\{b_{j},b_{j+1}\}\cup X_{\geq j+2}$.
  \item If $j = i+1$ then cutting the edge $u_Au_B$ separates $X_{\leq i-2} \cup \{b_{i-1},b_{i}\} = X_{\leq j-3} \cup \{b_{j-2},b_{j-1}\}$ from $\{a_{i+1},a_{i+2}\}\cup X_{\geq j+3} = \{a_{j},a_{j+1}\}\cup X_{\geq j+2}$.
  \item If $j = i+2$, then cutting $u_Bu_{i+1}$ separates $X_{\leq i-1} \cup \{a_{i},a_{i+1}\} = X_{\leq j-3} \cup \{a_{j-2},a_{j-1}\}$ from $\{b_{i+2},b_{i+3}\}\cup X_{\geq i+4} = \{b_{j},b_{j+1}\}\cup X_{\geq j+2}$.
  \end{itemize}
  
  Thus in each case, we have that $A_iB$ displays $\phi_j$.
\end{proof}

The next lemma follows from Lemmas~\ref{lem:AcrossBDisplaysOmegaAB},~\ref{lem:AcrossBDisplaysMostChi} and~\ref{lem:AcrossBDisplaysPhi}.

\begin{lemma}\label{lem:CminusChiCompatible}
 For any $i \in [n-1]$, lobster $A_iB$ is compatible with $C \setminus \{\chi_i\}$.
\end{lemma}

\subsubsection{Compatibility of $C\setminus \{\phi_i\}$ for each $3 \leq i \leq n-1$}\label{sec:compatibleNoPhi}

We now show that for any $3 \leq i \leq  n-1$, the set $C \setminus \{\phi_i\}$ is compatible.
To this end, we need to define a new type of tree $A^iB$, which we will show is compatible with $C \setminus \{\phi_i\}$.
This tree will be isomorphic to $A$ when restricted to  $X_{\leq i-1}$, and isomorphic to $B$ when restricted to $X_{\geq i}$.
This property is also true of $A_{i-1}B$, but the structure of $A^iB$ is slightly different.
(Fig.~\ref{fig:JaggedLobsters}).

\begin{definition}
For $2 \leq i \leq n-2$, 
the tree $A^iB$ on $X$ is defined as follows:
$A^iB$ has leaves $a_1, \dots, a_n, b_1, \dots, b_n$ and internal nodes $u_1, \dots, u_{i-2}$, $u_{i}, \dots, u_{n-1}$, $v_1, \dots, v_{i-2}$, $v_{i}, \dots,  v_{n-1}$, $u_A$, $u_B$
(note that $A_iB$ does not have vertices $u_{i-1}$ or $v_{i-1}$ but instead has $u_A$ and $u_B$).
 $A_iB$ contains a central path $a_1$,$u_1$,$u_2,\dots, u_{i-2}$,$u_A$, $u_B$,$u_{i}, \dots, u_{n-1}$,$b_n$.
For each $j \in [n-1] \setminus \{i-1\}$, there is an edge $u_jv_j$.
If $i$ is even then $u_A$ is adjacent to $a_i$ and $u_B$ is adjacent to $b_{i-1}$.
On the other hand if $i$ is odd then $u_A$ is adjacent to $b_i$ and $u_B$ is adjacent to $a_{i-1}$.
For $j < i-1$, the vertex $v_j$ is adjacent to  $b_j$ and $b_{j+1}$ if $j$ is odd, and adjacent to  $a_j$ and $a_{j+1}$ if $j$ is even.
For $j > i-1$, $v_j$ is adjacent to  $a_j$ and $a_{j+1}$ if $j$ is odd, and adjacent to  $b_j$ and $b_{j+1}$ if $j$ is even.
\end{definition}

\begin{figure}
\begin{center}
\begin{subfigure}{\textwidth}
 \centering\includegraphics[scale = 0.85]{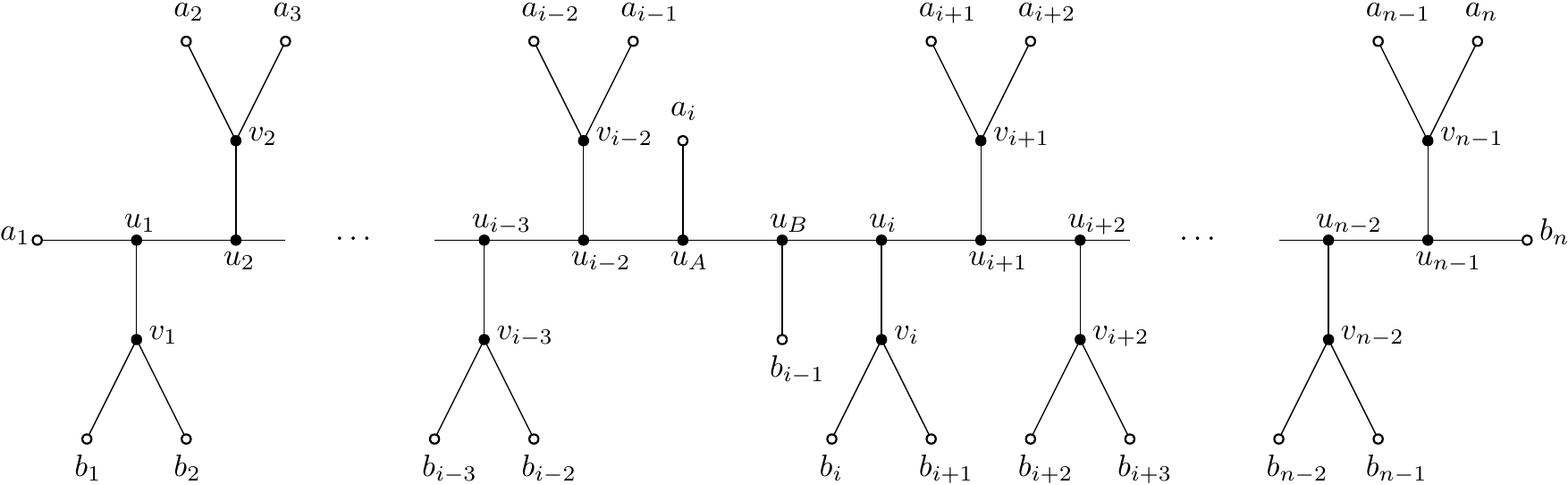}
 \caption{ Lobster $A^iB$ for $i$ even}    
\end{subfigure}

\vspace{1cm}

\begin{subfigure}{\textwidth}
 \centering\includegraphics[scale = 0.85]{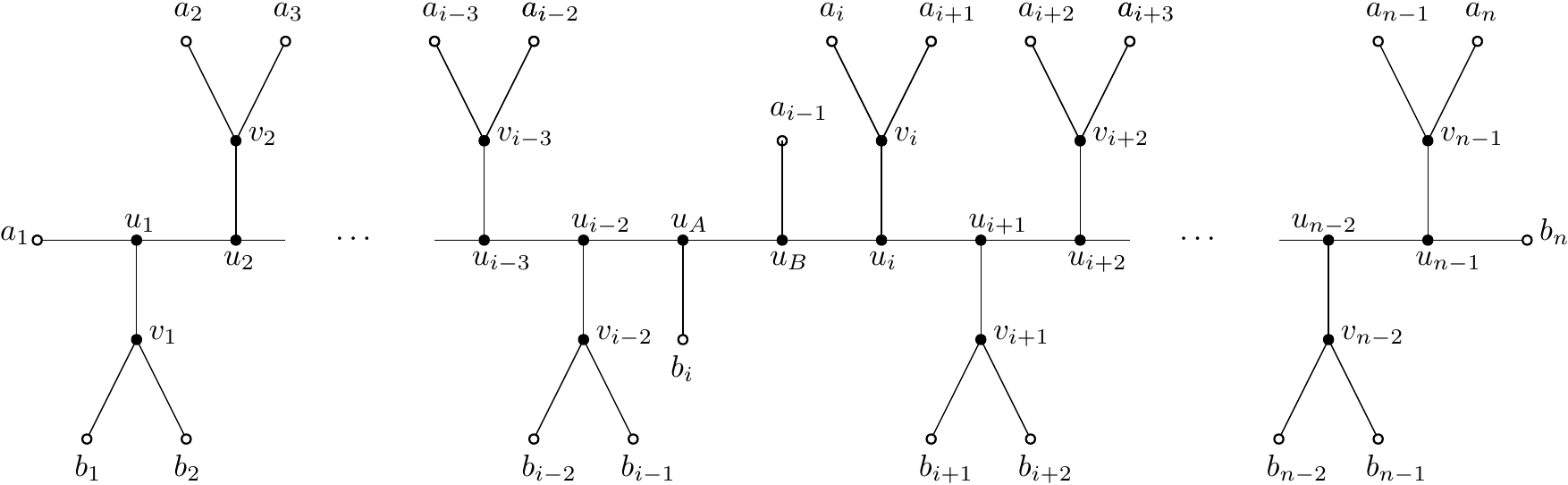}
 \caption{Lobster $A^iB$ for $i$ odd}    
\end{subfigure}
\end{center}
 \caption{The lobster $A^iB$, for the cases when $i$ is even and $i$ is odd.}\label{fig:JaggedLobsters}
\end{figure}

Observe that for $i$ even, $A^iB$ is equivalent to $A_{i-1}B$ with the leaves $a_i$ and $b_{i-1}$ swapped; 
for $i$ odd, $A^B$ is equivalent to $A_{i-1}B$ with the leaves $a_{i-1}$ and $b_{i}$ swapped.
We are now ready to show that $A^iB$ displays every character in $C$ except for $\phi_i$.

\begin{lemma}\label{lem:AjaggedBDisplaysOmegaAB}
 For any $3 \leq i \leq n-1$, lobster $A^iB$ displays $\Omega_A$ and  $\Omega_B$.
\end{lemma}
 \begin{proof}
 To see that $A^iB$ displays $\Omega_A$, observe that cutting the edge $u_1u_2$ (or $u_1u_A$ if $i=3$) separates $A_iB$ into two trees, one with leaves $a_1,b_1,b_2$ and one with leaf set $\{a_2\} \cup X_{\geq 3}$. It follows that the subtrees of $A^iB$ spanning these two sets are vertex-disjoint, and so $A^iB$ displays $\Omega_A$.
 Similarly, to see that $A^iB$ displays $\Omega_B$, observe that cutting the edge $u_{n-2}u_{n-1}$ (or $u_Bu_{n-1}$ if $i = n-1$) separates $A^iB$ into two trees, with leaf sets $X_{\leq n-2}\cup \{b_{n-1}\}$ and $\{a_{n-1}a_nb_n\}$ respectively.
 \end{proof}

\begin{lemma}\label{lem:AjaggedBDisplaysChi}
For any $3 \leq i \leq n-1$ and for each 
$2 \leq j \leq n-2$,
lobster $A^iB$ displays $\chi_j$.
\end{lemma}

\begin{proof}
The non-singleton states of $\chi_j$ are  $X_{\leq j-2}$, $\{a_j, a_{j+1}\}$, $\{b_j, b_{j+1}\}$ and $X_{\geq j+3}$.
 Cutting the edge  $u_{j-2}u_{j-1}$  ($u_{j-2}u_A$ if $j = i$, $u_Bu_{j-1}$ if $j= i +1$) separates $X_{\leq j-2}$ from the other non-singleton states.
 Similarly, cutting the edge $u_{j+1}u_{j+2}$ ($u_{j+1}u_A$ if $j = i -3$, $u_Bu_{j+2}$ if $j= i -2$) separates $X_{\leq j+3}$ from the other non-singleton states.
 
 It remains to show that the trees $A^iB[\{a_ja_{j+1}\}]$, $A^iB[\{b_j b_{j+1}\}]$ are vertex disjoint.
 For $j \neq i-1$, this can be seen by cutting the edge $u_jv_j$ (as $v_j$ is adjacent either to the leaves $a_i$ and $a_{i+1}$, or to the leaves $b_i$ and $b_{i+1}$).
 For $j = i-1$, this can be seen by cutting the edge $u_Au_B$.
\end{proof}

\begin{lemma}\label{lem:AjaggedBDisplaysMostPhi}
 For any $3 \leq i \leq n-1$ and for each $3 \leq j \leq n-1$ with $j \neq i$, lobster $A^iB$ displays $\phi_j$.
\end{lemma}

\begin{proof}
  The character $\phi_j$ has two non-singleton states; 
  these are either $X_{\leq j-3} \cup \{b_{j-2},b_{j-1}\}$ and $\{a_{j},a_{j+1}\}\cup X_{\geq j+2}$ (if $j$ is even)  or $X_{\leq j-3} \cup \{a_{j-2},a_{j-1}\}$ and $\{b_{j},b_{j+1}\}\cup X_{\geq j+2}$ (if $j$ is odd).
  
  We first consider the case when $j \notin \{i-1,i,i+1\}$
  In this case there are four possibilities to consider: 
  \begin{itemize}
  \item  If $j$ is even and $j < i-1$,
  then cutting the edge $u_{j-1}u_j$ separates $X_{\leq j-3} \cup \{b_{j-2},b_{j-1}\}$ from $\{a_{j},a_{j+1}\}\cup X_{\geq j+2}$. 
  \item If $j$ is even and $j > i+1$, then cutting the edge $u_{j-2}u_{j-1}$ separates $X_{\leq j-3} \cup \{b_{j-2},b_{j-1}\}$ from $\{a_{j},a_{j+1}\}\cup X_{\geq j+2}$.
  \item   If $j$ is odd and $j < i-1$, then cutting the edge $u_{j-1}u_j$ separates $X_{\leq j-3} \cup \{a_{j-2},a_{j-1}\}$ from $\{b_{j},b_{j+1}\}\cup X_{\geq j+2}$.
  \item  If $j$ is odd and $j > i+1$, then cutting the edge $u_{j-2}u_{j-1}$ separates $X_{\leq j-3} \cup \{a_{j-2},a_{j-1}\}$ from $\{b_{j},b_{j+1}\}\cup X_{\geq j+2}$.
  \end{itemize}

  We now consider the case when $j \in \{i-1,i+1\}$, and suppose first that $i$ is even (and thus $j$ is odd).
  \begin{itemize}
  \item  If $j = i-1$ then cutting the edge $u_Au_B$ separates $X_{\leq i-4} \cup \{a_{i-3},a_{i-2}\} = X_{\leq j-3} \cup \{a_{j-2},a_{j-1}\}$ from $\{b_{i-1},b_{i}\}\cup X_{\geq j+1} = \{b_{j},b_{j+1}\}\cup X_{\geq j+2}$.
  \item  If $j = i+1$ then again cutting the edge $u_Au_B$ separates $X_{\leq i-2} \cup \{a_{i-1},a_{i}\} = X_{\leq j-3} \cup \{a_{j-2},a_{j-1}\}$ from $\{b_{i+1},b_{i+2}\}\cup X_{\geq j+3} = \{b_{j},b_{j+1}\}\cup X_{\geq j+2}$.
  \end{itemize}

  Finally, consider the case when $j \in \{i-1,i+1\}$, and $i$ is odd (and thus $j$ is even).
  \begin{itemize}
  \item If $j = i-1$ then cutting the edge $u_Au_B$ separates $X_{\leq i-4} \cup \{b_{i-3},b_{i-2}\} = X_{\leq j-3} \cup \{b_{j-2},b_{j-1}\}$ from $\{a_{i-1},a_{i}\}\cup X_{\geq j+1} = \{a_{j},a_{j+1}\}\cup X_{\geq j+2}$.
  \item If $j = i+1$ then again cutting the edge $u_Au_B$ separates $X_{\leq i-2} \cup \{b_{i-1},b_{i}\} = X_{\leq j-3} \cup \{b_{j-2},b_{j-1}\}$ from $\{a_{i+1},a_{i+2}\}\cup X_{\geq j+3} = \{a_{j},a_{j+1}\}\cup X_{\geq j+2}$.
  \end{itemize}
  
  Thus in each case, we have that $A^iB$ displays $\phi_j$.
\end{proof}

The next lemma follows from Lemmas~\ref{lem:AjaggedBDisplaysOmegaAB},~\ref{lem:AjaggedBDisplaysChi} and~\ref{lem:AjaggedBDisplaysMostPhi}.

\begin{lemma}\label{lem:CminusPhiCompatible}
 For any $3 \leq i \leq n-1$, lobster $A^iB$ is compatible with $C \setminus \{\phi_i\}$.
\end{lemma}

Combining Lemmas~\ref{lem:CminusOmegaACompatible},~\ref{lem:CminusOmegaBCompatible},~\ref{lem:CminusChiCompatible} and~\ref{lem:CminusPhiCompatible}, we have the following lemma.

\begin{lemma}\label{lem:subsetConsistency}
 For any $C' \subseteq C$ with $C' \neq C$, $C'$ is compatible.
\end{lemma}

\subsubsection{Incompatibility of $C$}\label{sec:incompatible}



Let $S_1|S_2|S_3|S_4$ be a partition of $X' \subseteq X$.
We say a tree $T$ on $X$ displays the quartet $S_1 | S_2 \parallel  S_3|S_4$ 
if there exist internal vertices $u$ and $v$, such that deleting any edge on the path from $u$ to $v$
separates $S_1 \cup S_2$ from $S_3 \cup S_4$, and in addition deleting $u$ separates $S_1$ from $S_2$, and deleting $v$ separates $S_3$ from $S_4$
(Fig.~\ref{fig:ExampleQuartet}). Note that this notion is a generalization of the usual notion of displaying a quartet, in which each of the sets~$S_1,\ldots ,S_4$ consists of a single leaf.

\begin{figure}
\begin{center}
 \centering\includegraphics[scale = 1]{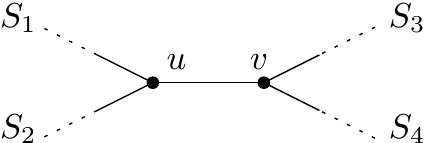}
\end{center}
 \caption{A tree displaying the quartet $S_1|S_2 \parallel  S_3| S_4$.}\label{fig:ExampleQuartet}
\end{figure}

\begin{definition}\label{def:leafMeetsTree}
 Given a tree $T$ on $X$ and a leaf $x \notin X'\subseteq X$, 
 we say that $x$ \emph{meets} $T[X']$ at a vertex $v$  if $v$ is a vertex in $T[X']$ and there is a path from $v$ to $x$ in $T$ that is edge-disjoint from $T[X']$.
 We say $x$ \emph{meets $T[X']$ between $u$ and $v$} if $u,v$ are two vertices in $T[X']$ and $x$ meets $T[X']$ at $v'$ for some vertex $v'$ on the path from $u$ to $v$.
\end{definition}

To prove that $C$ is incompatible, we will prove that any tree compatible with $C \setminus \{\Omega_B\}$ must display certain quartets.
In particular, it must display a quartet that cannot be displayed by a tree displaying $\Omega_B$. This implies that there is no tree compatible with $C$.
The next lemma gives the base case and the following two lemmas give the inductive step of this proof.


 \begin{lemma}\label{lem:incompatibleBaseCase}
If $T$ is a tree on $X$ that displays $\Omega_A$, $\chi_2$, $\phi_3$, $\chi_3$ and $\chi_4$, then 
$T$ displays $X_{\leq 2} \cup \{a_{3}\}|a_{4}\parallel b_{3}|b_{4}$.
\end{lemma}
\begin{proof}

Let $u_1$ be the vertex in $T$ at which $a_2$ joins the subtree $T[\{a_1,b_1,b_2\}]$.
Let $u_2$ be the vertex at which $b_3$ joins $T[\{a_1,b_1,a_2,b_2\}]$.
Observe that since $T$ displays $\Omega_A$, $u_2$ must be between $u_1$ and $a_2$.
Indeed, if this is not the case then the path from $a_2$ to $b_3$ must pass through $u_1$, which is also part of the subtree $T[\{a_1,b_1,b_2\}]$, contradicting the fact that $T[\{a_1,b_1,b_2\}]$ and $T[\{a_2\} \cup X_{\geq 3}]$ are vertex-disjoint
(Fig.~\ref{fig:Tupto2plusb3}).

Now let $v_2$ be the vertex at which $a_3$ joins $T[\{a_1,b_1,a_2,b_2, b_3\}]$.
As $T$ displays $\chi_2$, $v_2$ must be between $u_2$ and $a_2$, since otherwise the subtrees $T[\{a_2,a_3\}]$ and $T[\{b_2,b_3\}]$ both contain $u_2$.
Next let $v_3$ be the vertex at which $b_4$ joins $T[X_{\leq 3}]$.
As $T$ displays $\phi_3$, $v_3$ must be between $u_2$ and $b_3$, since otherwise the subtrees $T[\{a_1,a_2\}]$ and $T[\{b_3,b_4\} \cup X_{\geq 5}]$ both contain $u_2$
(Fig.~\ref{fig:adda3b4}).

Now in order to show that $T$ displays $X_{\leq 2} \cup \{a_{3}\}|a_{4}\parallel b_{3}|b_{4}$,
it remains to determine the relative poition of $a_4$.
In order to do this we need to consider $a_5$, although we will not determine 
the position of $a_5$ itself.
As $T$ displays $\phi_3$, the subtrees $T[\{a_1,a_2\}]$ and $T[\{b_3,b_4\} \cup X_{\geq 5}]$ are vertex-disjoint,
and in particular the path from $b_4$ to $a_5$ must not contain $u_2$.
Also as $T$ displays $\chi_4$ (and thus $T[X_{\leq 2}]$ and $T[\{a_4,a_5\}]$ are vertex-disjoint), the path from $a_4$ to $a_5$ does not contain $u_2$.
As neither of the paths $T[\{b_4,a_5\}]$ and $T[\{a_4,a_5\}]$ contain $u_2$, it follows that the path $T[\{a_4,b_4\}]$ does not contain $u_2$ either (note that the path $T[\{a_4,b_4\}]$ is a subgraph of the union of $T[\{b_4,a_5\}]$ and $T[\{a_4,a_5\}]$).
This implies that $a_4$ meets $T[X_{\leq 3} \cup \{b_4\}]$ in one of three-places: either between $v_3$ and $b_4$, between $v_3$ and $b_3$, or between $u_2$ and $v_3$.
However, as $T$ displays $\chi_3$ (and thus $T[\{a_3,a_4\}]$ and $T[\{b_3,b_4\}]$ are vertex-disjoint), 
the path $T[\{a_3,a_4\}]$ cannot contain $v_3$. 
This implies that $a_4$ must meet $T[X_{\leq 3} \cup \{b_4\}]$ between $u_2$ and $v_3$.
Let $u_3$ be the vertex at which  $a_4$ meets $T[X_{\leq 3} \cup \{b_4\}]$
(Fig~\ref{fig:adda4}).

Now observe that deleting the edge $u_3v_3$ separates $X_{\leq 2} \cup \{a_3,a_4\}$ from $\{b_3,b_4\}$, that deleting $u_3$ separates $X_{\leq 2} \cup \{a_3\}$ from $a_4$, 
and that deleting $v_3$ separates $b_3$ from $b_4$.
Thus, $T$ displays  $X_{\leq 2} \cup \{a_{3}\}|a_{4}\parallel b_{3}|b_{4}$.
\end{proof}

\begin{figure}
\begin{subfigure}{0.25\textwidth}
 \includegraphics[scale = 0.85]{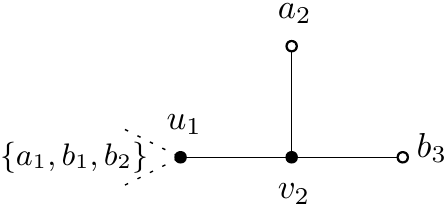}
 \caption{$T[\{a_1, b_1, a_2, b_2, b_3\}]$}\label{fig:Tupto2plusb3}    
\end{subfigure}
~
\begin{subfigure}{0.3\textwidth}
 \includegraphics[scale = 0.85]{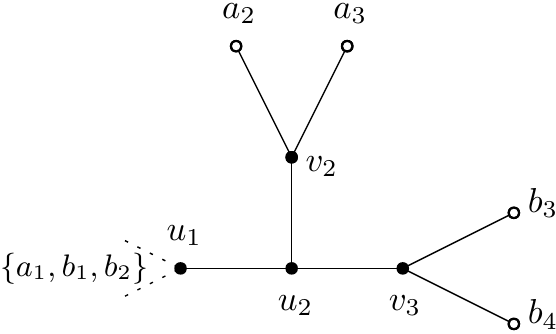}
 \caption{ $T[X_{\leq {3}}\cup \{b_{4}\}]$}\label{fig:adda3b4}     
\end{subfigure}
~
\begin{subfigure}{0.3\textwidth}
 \includegraphics[scale = 0.85]{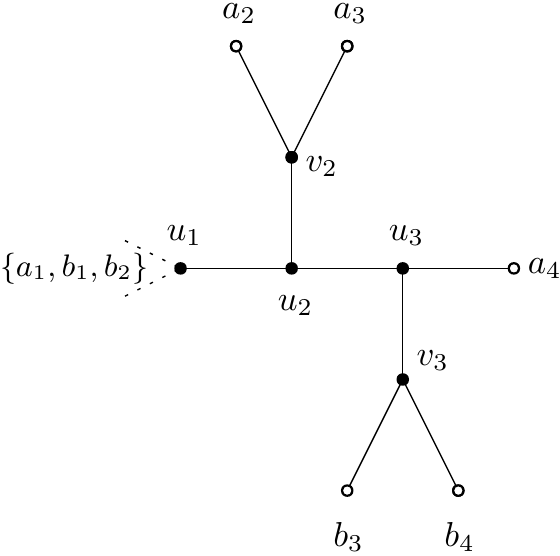}
 \caption{ $T[X_{\leq {4}}]$}\label{fig:adda4}    
\end{subfigure}
 \caption{Illustration of the proof of Lemma~\ref{lem:incompatibleBaseCase}. Trees are drawn with their degree-$2$ vertices suppressed.}
 \label{fig:caterpillar}
\end{figure}

%
%

\begin{lemma}\label{lem:incompatibleInductionStepEven}
 Let $i \in [n-2]$ such that $i \geq 4$ and $i$ is even.
 If $T$ is a tree on $X$ such that 
 $T$ displays $X_{\leq i-2} \cup \{a_{i-1}\}|a_i\parallel b_{i-1}|b_i$
 and 
 $T$ displays $\chi_{i-2}, \chi_i$ and $\phi_i$,
 then 
 $T$ displays $X_{\leq i-1} \cup \{b_{i}\}|b_{i+1}\parallel a_{i}|a_{i+1}$.
\end{lemma}
\begin{proof}
 Let $u_{i-1}, v_{i-1}$ be internal vertices in $T$ such that deleting any edge on the path from $u_{i-1}$ to $v_{i-1}$ separates $X_{\leq i-2} \cup \{a_{i-1},a_i\}$ from $\{b_{i-1},b_i\}$, deleting $u_{i-1}$ separates $X_{\leq i-2} \cup \{a_{i-1}\}$ from $\{a_i\}$, and deleting $v_{i-1}$ separates $b_{i-1}$ from $b_i$ (Fig.~\ref{fig:Tuptoi}).
 
 As $T$ displays $\phi_i$, it must be that $a_{i+1}$ meets $T[X_{\leq i}]$ between $u_{i-1}$ and $a_i$, as otherwise the subtrees $T[X_{\leq i-3}\cup \{b_{i-2},b_{i-1}\}]$ and $T[\{a_i,a_{i+1}\} \cup X_{i \geq 2}]$ are not vertex-disjoint (in particular, the paths $T[\{b_{i-2},b_{i-1}\}]$ and  $T[\{a_i,a_{i+1}\}]$ both contain $u_{i-1}$).
 Let $v_i$ be the vertex at which $a_{i+1}$ meets $T[X_{\leq i}]$ (Fig.~\ref{fig:addaiplus}).
 
 Now consider $b_{i+1}$. As $T$ displays $\chi_{i-2}$, the paths $T[\{b_{i-2},b_{i-1}\}]$ and $T[\{a_{i+1},b_{i+1}\}]$ are vertex-disjoint.
 It follows that $T[\{a_{i+1},b_{i+1}\}]$  cannot contain $u_{i-1}$, and so $b_{i+1}$ joins $T[X_{\leq i} \cup \{a_{i+1}\}]$ at one of three places: either between $v_i$ and $a_{i+1}$, between $v_i$ and $a_i$, or between $u_{i-1}$ and $v_i$.
 Furthermore as $T$ displays $\chi_i$, the paths $T[\{a_i,a_{i+1}\}]$ and $T[\{b_i,b_{i+1}\}]$ are vertex disjoint, and in particular $T[\{b_i,b_{i+1}\}]$ cannot contain $v_i$.
 It follows that $b_{i+1}$ joins $T[X_{\leq i} \cup \{a_{i+1}\}]$ between $u_{i-1}$ and $v_i$.
 Let $u_i$ be the vertex at which  $b_{i+1}$ joins $T[X_{\leq i} \cup \{a_{i+1}\}]$ (Fig.~\ref{fig:addbiplus}).
 
 Now observe that deleting $u_iv_i$ separates $X_{\leq i-1} \cup \{b_{i},b_{i+1}\}$ from $\{a_{i},a_{i+1}\}$,
 that deleting $u_i$ separates $X_{\leq i-1} \cup \{b_{i}\}$ from $b_{i+1}$,
 and that deleting $v_i$ separates $a_i$ from $a_{i+1}$.
Thus, $T$ displays $X_{\leq i-1} \cup \{b_{i}\}|b_{i+1}\parallel a_{i}|a_{i+1}$.
\end{proof}

\begin{figure}
\begin{subfigure}{0.25\textwidth}
 \includegraphics[scale = 0.85]{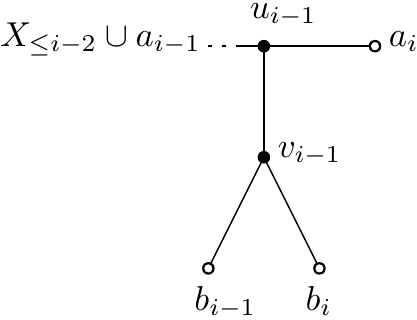}
 \caption{  $T[X_{\leq {i}}]$}\label{fig:Tuptoi}    
\end{subfigure}
~
\begin{subfigure}{0.3\textwidth}
 \includegraphics[scale = 0.85]{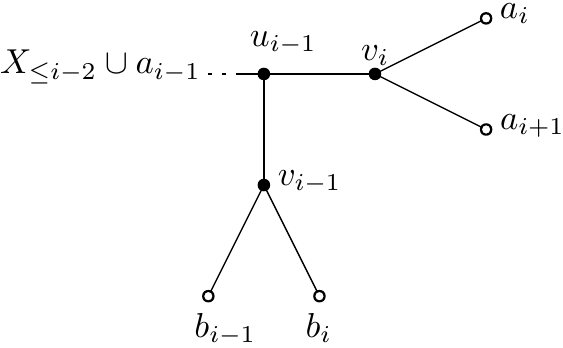}
 \caption{ $T[X_{\leq {i}}\cup \{a_{i+1}\}]$}\label{fig:addaiplus}     
\end{subfigure}
~
\begin{subfigure}{0.3\textwidth}
 \includegraphics[scale = 0.85]{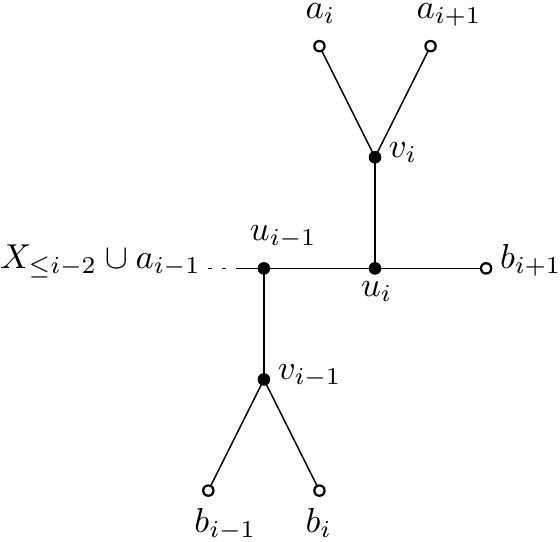}
 \caption{ $T[X_{\leq {i+1}}]$}\label{fig:addbiplus}    
\end{subfigure}
 \caption{Illustration of the proof of Lemma~\ref{lem:incompatibleInductionStepEven}. Trees are drawn with their degree-$2$ vertices suppressed.}
 \label{fig:caterpillar}
\end{figure}

\begin{lemma}\label{lem:incompatibleInductionStepOdd}
 Let $i \in [n-2]$ such that $i > 4$ and $i$ is odd.
 If $T$ is a tree on $X$ such that 
 $T$ displays $X_{\leq i-2} \cup \{b_{i-1}\}|b_i\parallel a_{i-1}|a_i$
 and 
 $T$ displays $\chi_{i-2}, \chi_i$ and $\phi_i$,
 then 
 $T$ displays $X_{\leq i-1} \cup \{a_{i}\}|a_{i+1}\parallel b_{i}|b_{i+1}$.
\end{lemma}
\begin{proof}
The proof is symmetric to that of Lemma~\ref{lem:incompatibleInductionStepEven}.
 Let $u_{i-1}, v_{i-1}$ be internal vertices in $T$ such that deleting any edge on the path from $u_{i-1}$ to $v_{i-1}$ separates $X_{\leq i-2} \cup \{b_{i-1},b_i\}$ from $\{a_{i-1},a_i\}$, deleting $u_{i-1}$ separates $X_{\leq i-2} \cup \{b_{i-1}\}$ from $\{b_i\}$, and deleting $v_{i-1}$ separates $a_{i-1}$ from $a_i$ (Fig.~\ref{fig:TuptoiOdd}).
 
 As $T$ displays $\phi_i$, it must be that $b_{i+1}$ meets $T[X_{\leq i}]$ between $u_{i-1}$ and $b_i$, as otherwise the subtrees $T[X_{\leq i-3}\cup \{a_{i-2},a_{i-1}\}]$ and $T[\{b_i,b_{i+1}\} \cup X_{i \geq 2}]$ are not edge-disjoint (in particular, the paths $T[\{a_{i-2},a_{i-1}\}]$ and  $T[\{b_i,b_{i+1}\}]$both contain $u_{i-1}$).
 Let $v_i$ be the vertex at which $b_{i+1}$ meets $T[X_{\leq i}]$  (Fig.~\ref{fig:addaiplusOdd}).
 
  Now consider $a_{i+1}$. As $T$ displays $\chi_{i-2}$, the paths $T[\{a_{i-2},a_{i-1}\}]$ and $T[\{a_{i+1},b_{i+1}\}]$ are vertex-disjoint.
 It follows that $a_{i+1}$ joins $T[X_{\leq i} \cup \{b_{i+1}\}]$ at one of three places: either between $v_i$ and $b_{i+1}$, between $v_i$ and $b_i$, or between $u_{i-1}$ and $v_i$.
 Furthermore as $T$ displays $\chi_i$, the paths $T[\{b_i,b_{i+1}\}]$ and $T[\{a_i,a_{i+1}\}]$ are vertex disjoint, and in particular $T[\{a_i,a_{i+1}\}]$ cannot contain $v_i$.
 It follows that $a_{i+1}$ joins $T[X_{\leq i} \cup \{b_{i+1}\}]$ between $u_{i-1}$ and $v_i$.
 Let $u_i$ be the vertex at which  $a_{i+1}$ joins $T[X_{\leq i} \cup \{b_{i+1}\}]$ (Fig.~\ref{fig:addbiplusOdd}).
 
 Now observe that deleting $u_iv_i$ separates $X_{\leq i-1} \cup \{a_{i},a_{i+1}\}$ from $\{b_{i},b_{i+1}\}$,
 that deleting $u_i$ separates $X_{\leq i-1} \cup \{a_{i}\}$ from $a_{i+1}$,
 and that deleting $v_i$ separates $b_i$ from $b_{i+1}$.
Thus  $T$ displays $X_{\leq i-1} \cup \{a_{i}\}|a_{i+1}\parallel b_{i}|b_{i+1}$.
\end{proof}

\begin{figure}
\begin{subfigure}{0.25\textwidth}
 \includegraphics[scale = 0.8]{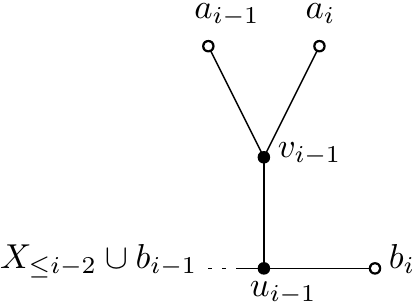}
 \caption{  $T[X_{\leq {i}}]$}\label{fig:TuptoiOdd}    
\end{subfigure}
~
\begin{subfigure}{0.3\textwidth}
 \includegraphics[scale = 0.8]{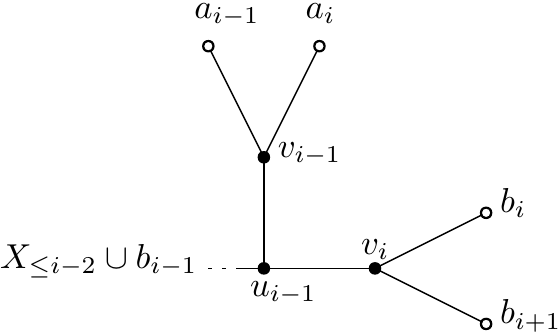}
 \caption{ $T[X_{\leq {i}}\cup \{a_{i+1}\}]$}\label{fig:addaiplusOdd}     
\end{subfigure}
~
\begin{subfigure}{0.3\textwidth}
 \includegraphics[scale = 0.8]{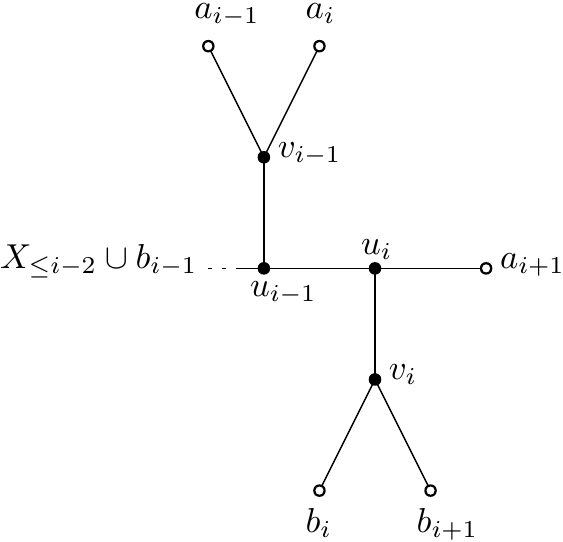}
 \caption{ $T[X_{\leq {i+1}}]$}\label{fig:addbiplusOdd}    
\end{subfigure}
 \caption{Illustration of the proof of Lemma~\ref{lem:incompatibleInductionStepOdd}. Trees are drawn with their degree-$2$ vertices suppressed.}
\end{figure}


\begin{lemma}\label{lem:AorBust}
 For any $3 \leq i \leq n-2$, if a tree $T$ is compatible with $\Omega_A$ and with $\chi_j$ for all $2 \leq j \leq i$ and $\phi_j$ for all $3 \leq j \leq i$, then 
 $T$ displays $X_{\leq i-1} \cup \{b_{i}\}|b_{i+1}\parallel a_{i}|a_{i+1}$ if $i$ is even,
 and 
 $T$ displays $X_{\leq i-1} \cup \{a_{i}\}|a_{i+1}\parallel b_{i}|b_{i+1}$ if $i$ is odd.
In particular if $T$ is compatible with $C \setminus \{\phi_{n-1},\Omega_B\}$ then  $T$ displays $X_{\leq n-3} \cup \{b_{n-2}\}|b_{n-1}\parallel a_{n-2}|a_{n-1}$.
\end{lemma}
\begin{proof}
The claim follows by induction on $i$. 
 For $i = 3$, the claim follows from Lemma~\ref{lem:incompatibleBaseCase}.
 For larger values of $i$, if $i$ is even then the claim follows from Lemma~\ref{lem:incompatibleInductionStepEven} and the fact that the claim holds for $i-1$.
 If $i$ is odd, the claim follows from Lemma~\ref{lem:incompatibleInductionStepOdd} and the fact that the claim holds for $i-1$.
\end{proof}

\begin{lemma}\label{lem:OmegaBdoesNotPlayNicelyWithQuartet}
 If $T$ is a tree on $X$ such that 
 $T$ displays $X_{\leq n-3} \cup \{b_{n-2}\}|b_{n-1}\parallel a_{n-2}|a_{n-1}$,
 then either $T$ does not display $\phi_{n-1}$ or $T$ does not display $\Omega_B$.
\end{lemma}
\begin{proof}
 Let $u_{n-2}, v_{n-2}$ be internal vertices in $T$ such that deleting any edge on the path from $u_{n-2}$ to $v_{n-2}$ separates $X_{\leq n-3} \cup \{b_{n-2},b_{n-1}\}$ from $\{a_{n-2},a_{n-1}\}$, deleting $u_{n-2}$ separates $X_{\leq n-3} \cup \{b_{n-2}\}$ from $\{b_{n-1}\}$, and deleting $v_{n-2}$ separates $a_{n-2}$ from $a_{n-1}$ (Fig.~\ref{fig:QuartetOmegaBIncompatibility}).
 
 If $T$ displays  $\phi_{n-1}$, then the subtrees $T[X_{\leq n-4} \cup \{a_{n-3},a_{n-2}\}]$ and $T[\{b_{n-1},b_n\}]$ are vertex-disjoint, and in particular the path $T[\{b_{n-1},b_n\}]$ does not contain $u_{n-2}$.
 It follows that $b_n$ joins $T[X_{\leq n-1}]$ between $u_{n-2}$ and $b_{n-1}$.
 On the other hand, if $T$ displays $\Omega_B$, then the subtrees $T[X_{\leq n-2} \cup \{b_{n-1}\}]$ and $T[\{a_{n-1},a_n,b_n\}]$ are vertex-disjoint, and in particular the path $T[\{a_{n-1},b_n\}]$ does not contain $v_{n-2}$.
 It follows that $b_n$ joins $T[X_{\leq n-1}]$ between $v_{n-2}$ and $a_{n-1}$.
 As $b_n$ cannot join  $T[X_{\leq n-1}]$ in two different locations, $T$ either does not display $\phi_{n-1}$ or does not display $\Omega_B$.
\end{proof}

\begin{figure}
\begin{center}
 \centering\includegraphics[scale = 1]{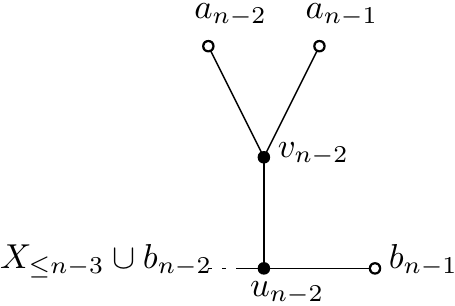}
\end{center}
 \caption{Illustration of the proof of Lemma~\ref{lem:OmegaBdoesNotPlayNicelyWithQuartet}.}\label{fig:QuartetOmegaBIncompatibility}
\end{figure}

\begin{lemma}\label{lem:incompatible}
 $C$ is not compatible.
\end{lemma}
\begin{proof}
 This follows immediately from Lemmas~\ref{lem:AorBust} and~\ref{lem:OmegaBdoesNotPlayNicelyWithQuartet}.
\end{proof}

By choosing $n$ such that $2n-4 > t$, Lemmas~\ref{lem:subsetConsistency} and~\ref{lem:incompatible} give us the following theorem, which shows that Conjecture~\ref{con:perfectPhylogenyCharacterization} is false.

\begin{theorem}
 For any integer $t$, there exists a set $C$ of $8$-state characters such that $C$ is incompatible but  every subset of at most $t$ characters in $C$ is compatible.
\end{theorem}

\section{Discussion}

First note that we have only described a counter example for the case that there are~$2n$ taxa with~$n\geq 4$ even. However, we can easily create examples for any number of taxa, that is at least~8, by ``copying'' taxa. More precisely, we can replace, say,~$a_1$ by any number of taxa that all have the same state as~$a_1$ in all characters.

Secondly, we describe how our counter example can be seen as a counter example with four different states and gaps. Considering Definition~\ref{def:characters}, observe that each of the characters has at most four states that contain more than one taxon. The remaining states contain just one taxon and can therefore be replaced by gaps (indicating that we do not know which state the taxon has in that character). This gives a counter example with four different states and gaps. One can argue that the local obstruction conjecture is anyway unlikely to be true even for binary characters with gaps, because if it were true we would then be able to solve the quartet compatibility problem (see, e.g.,~\citet{SempleSteel2003}) in polynomial time, which would in turn imply that the complexity classes P and NP would coincide. However, one appealing feature of our counter example is that it does not rely on any assumptions on complexity classes.

We conclude the paper by reiterating that, if we do not allow gaps, the local obstructions conjecture restricted to characters with~$4,5,6$ or~$7$ states is still open.

\section{Funding}
 This work was supported by the Netherlands Organization for Scientific Research (including Vidi grant 639.072.602), and by the 4TU Applied Mathematics Institute.


\end{document}